\documentclass{article}

\usepackage{graphicx} 
\usepackage{amsfonts}
\usepackage{amsmath}
\usepackage{mathrsfs}
\usepackage{amssymb}
\usepackage{url}
\usepackage{fancyhdr}
\usepackage{indentfirst}
\usepackage{enumerate}
\usepackage{cases}
\usepackage{dsfont,footmisc}

\newcommand\smallO{
  \mathchoice
    {{\scriptstyle\mathcal{O}}}
    {{\scriptstyle\mathcal{O}}}
    {{\scriptscriptstyle\mathcal{O}}}
    {\scalebox{.7}{$\scriptscriptstyle\mathcal{O}$}}
  }

\usepackage{subcaption}
\usepackage{amsthm}
\usepackage{color}
\usepackage{natbib}
\usepackage[british]{babel}
\usepackage[colorlinks=true,citecolor=blue]{hyperref}

\usepackage[right,pagewise,displaymath, mathlines]{lineno}

\usepackage[font=scriptsize]{caption} 
\usepackage{soul,xcolor}
\setstcolor{red}

\usepackage[title]{appendix}
\def\d{\mathrm{d}}
\def\laweq{\buildrel \d \over =}
\def\lawto{\buildrel \d \over \longrightarrow}

\newcommand{\var}{\mathrm{Var}}
\newcommand{\cov}{\mathrm{Cov}}

\renewcommand{\O}{\mathcal{O}}

\newcommand{\VaR}{\mathrm{VaR}}
\newcommand{\ex}{\mathrm{ex}}

\newcommand{\ES}{\mathrm{ES}}
\newcommand{\DQVaR}{\mathrm{DQ}^{\VaR}_\alpha}
\newcommand{\DQrho}{\mathrm{DQ}^{\rho}_\alpha}
\newcommand{\DQES}{\mathrm{DQ}^{\ES}_\alpha}
\newcommand{\DQex}{\mathrm{DQ}^{\ex}_\alpha}
\newcommand{\wDQVaR}{\widehat{\mathrm{DQ}^{\VaR}_\alpha}}
\newcommand{\wDQrho}{\widehat{\mathrm{DQ}^{\rho}_\alpha}}
\newcommand{\wDQES}{\widehat{\mathrm{DQ}^{\ES}_\alpha}}
\newcommand{\wDQex}{\widehat{\mathrm{DQ}^{\ex}_\alpha}}

\newcommand{\wF}{\widehat{F}}
\newcommand{\wG}{\widehat{G}}
\newcommand{\E}{\mathbb{E}}
\newcommand{\R}{\mathbb{R}}
\newcommand{\bX}{\mathbf{X}}

\renewcommand{\k}{{(k)}}

\newcommand{\N}{\mathbb{N}}
\newcommand{\p}{\mathbb{P}}

\newcommand{\id}{\mathds{1}}

\newcommand{\X}{\mathcal X}

\newcommand{\esssup}{\mathrm{ess\mbox{-}sup}}

\renewcommand{\ge}{\geqslant}
\renewcommand{\le}{\leqslant}
\renewcommand{\geq}{\geqslant}
\renewcommand{\leq}{\leqslant}
\renewcommand{\epsilon}{\varepsilon}

\theoremstyle{plain}
\newtheorem{theorem}{Theorem}
\newtheorem{lemma}{Lemma}
\newtheorem{proposition}{Proposition}

\theoremstyle{definition}
\newtheorem{definition}{Definition}
\newtheorem{assumption}{Assumption}

\theoremstyle{remark}
\newtheorem{remark}{Remark}

\theoremstyle{definition}

\renewcommand{\cite}{\citet}
\renewcommand{\cdots}{\dots}
\DeclareMathOperator*{\argmin}{arg\,min}
\usepackage[onehalfspacing]{setspace}

\usepackage{geometry}

\usepackage{float}[H]

\begin{document}
\title{Empirical estimator of  diversification quotient}

\author{
Xia Han\thanks{School of Mathematical Sciences,  LPMC and AAIS,  Nankai University,  China.
			\texttt{xiahan@nankai.edu.cn}}
\and Liyuan Lin\thanks{Department of Econometrics and Business Statistics, Monash University,  Australia. \texttt{liyuan.lin@monash.edu}}
\and 	Mengshi Zhao\thanks{(Corresponding author) School of Mathematical Sciences, Nankai University, China. \texttt{2120230079@mail.nankai.edu.cn}}
}
\date{\today}

\maketitle

\begin{abstract}

The diversification quotient (DQ), introduced by \cite{HLW25}, is a new measure of portfolio diversification that quantifies the reduction in a portfolio’s security level attributable to diversification. Built on a rigorous axiomatic framework, DQ effectively captures heavy tails and common shocks while improving efficiency in portfolio optimization. This paper investigates the consistency and asymptotic normality of empirical DQ estimators based on Value at Risk (VaR) and Expected Shortfall (ES), and provides explicit formulas for the asymptotic variance under both independent and dependent data. Compared with the classical diversification ratio of \cite{T07}, which has diverging asymptotic variance due to its lack of location invariance, the DQ estimators exhibit greater robustness across diverse distributional settings. We further verify their asymptotic properties under elliptical distributions through simulation, and construct confidence intervals for DQ estimates using AR-GARCH models with a residual-based bootstrap on real financial data. These results establish a solid statistical foundation for applying DQ in financial risk management and decision-making.
\medskip
\noindent

\end{abstract}

\section{Introduction}

  Diversification is a fundamental principle in portfolio management, designed to reduce risk by combining assets with different risk profiles. Several indices have been proposed to quantify diversification such as the Diversification Ratio (DR, e.g., \cite{T07}, \cite{BDI08}, \cite{EWW15}) and Diversification Benefit (DB, e.g., \cite{EFK09}, \cite{MFE15}), focusing on evaluate the change in capital requirements that arises from asset pooling. However, these measures face notable limitations: they ignore the effect of heavy tails and common shocks, can be easily manipulated through shifts in constants or scaling, and present difficulties when applied to credit risk estimation.
To address these limitations, the diversification quotient (DQ), recently introduced by \cite{HLW25}, provides a promising alternative. Defined through a parametric class of risk measures, DQ quantifies the reduction in a portfolio’s security level attributable to diversification. DQ is the first diversification index built on a rigorous axiomatization framework, characterized by six axioms: non-negativity, location invariance, scale invariance, rationality, normalization, and continuity. These axioms distinguish DQ from existing diversification measures, ensuring both robustness and reliability across a wide range of financial applications.
 
Value at Risk (VaR) and Expected Shortfall (ES) families are widely used in financial risk management, particularly within regulatory frameworks such as Basel III/IV (\cite{BASEL19}) and Solvency II (\cite{E11}). The DQ measures based on VaR and ES have shown strong advantages in capturing heavy tails and common shocks, as well as in improving portfolio optimization efficiency (\cite{HLW23, HLW25}). However, the lack of well-established statistical properties for DQ estimators has limited their empirical application in finance and economics. This paper seeks to overcome this barrier by discussing how to estimate the value of DQ.

Supported by the fact that DQ based on law-invariant risk measures inherits the property of law-invariance (Proposition \ref{prop:1}), we naturally study the empirical estimator of DQ, defined as the value of DQ computed from the empirical distribution of portfolio losses. Compared with parametric estimators, the empirical approach is model-free and thus less vulnerable to distributional misspecification due to the lack of extreme event data.
Therefore, the nonparametric method naturally captures fat tails and asymmetries present in financial returns.
 This feature is particularly important in risk management, as underestimating tail risks can lead to significant financial losses and suboptimal portfolio allocation. 
 These advantages have attracted growing attention in estimating VaR and ES. The earliest work on empirical VaR can be traced back to \cite{B66}, which established consistency and asymptotic normality of the empirical VaR estimator under i.i.d.~data. For Expected Shortfall (ES), \cite{S04} introduced a nonparametric kernel estimator and applied it to sensitivity analysis in portfolio allocation. Extensions to dependent data include \cite{CT05}, who explored nonparametric inference for VaR, and \cite{C08}, who proposed two nonparametric ES estimators for dependent losses—one based on the sample mean of excess losses beyond VaR, and another as its kernel-smoothed version. More recently, nonparametric estimation techniques have been extended to broader classes of risk measures induced by VaR and ES (\cite{APWY19}; \cite{BFWW22}; \cite{LW23}).
 Together, these developments provide a natural statistical foundation for constructing empirical DQ estimators based on VaR and ES, enabling robust and distribution-free assessment of portfolio diversification.


We first establish the consistency of the DQ estimator for VaR and for classes of convex risk measures (including ES) under mild continuity conditions, ensuring that the estimator converges to the true value as the empirical distribution converges (Theorem \ref{thm:1}). Building on this result, we derive the asymptotic normality of the empirical DQ estimators based on VaR and ES and explicitly compute their asymptotic variances in the i.i.d.~setting (Theorems \ref{thm:2} and \ref{thm:3}). These results provide a rigorous statistical foundation for evaluating the variability and reliability of DQ estimates, which is critical for portfolio managers when assessing diversification benefits and exposure to tail risk.


Since financial returns typically exhibit temporal dependence, we extend the asymptotic normality results to $\alpha$-mixing sequences. In this more general dependent setting, the estimators remain both consistent and asymptotically normal (Theorems \ref{thm:4} and \ref{thm:5}), ensuring robustness in the presence of market shocks, volatility clustering, and serial correlation. Consequently, DQ-based assessments of risk reduction and capital allocation remain reliable even in complex, dependent market environments. We further analyze the sensitivity of the asymptotic variance of DQ estimators to changes in confidence level, correlation structure, portfolio dimension, and tail heaviness (degrees of freedom) using elliptically distributed portfolios, including multivariate normal and t-distributions. Elliptical distributions provide a fundamental framework in quantitative risk management (\cite{MFE15}).

We also establish the asymptotic normality and derive the asymptotic variance of the DR based on VaR and ES (Proposition \ref{prop:5}). To the best of our knowledge, these theoretical results for DR have not been previously studied. It is important to note that DR lacks location invariance, which may lead to potential instability in practice. Specifically, under certain shifts in the location of the portfolio returns, the asymptotic variance of DR can diverge, making its estimates highly sensitive to changes in the underlying data. This sensitivity limits the reliability of DR in risk assessment and portfolio evaluation, particularly when extreme events or heavy-tailed distributions are present. In contrast, DQ is explicitly designed to be location-invariant, ensuring that its risk-adjusted diversification measure remains stable even when the portfolio experiences location shifts.  

The remainder of the paper is organized as follows. Section \ref{sec:2} reviews the definition of the DQ and introduces its empirical estimator. Section \ref{sec:3} establishes the consistency of the empirical DQ estimators for a general class of risk measures. In Section \ref{sec:4}, we investigate the asymptotic normality for the empirical estimators of DQ based on VaR and ES classes. Section \ref{sec:5} extends these results to settings with dependent data, while Section \ref{sec:6} develops the corresponding asymptotic theory for the DR. Section \ref{sec:7} provides numerical illustrations that demonstrate the behavior of DQ under elliptical distributional assumptions. In Section \ref{sec:real_data}, we apply our methodology to real financial data. Finally, Section \ref{sec:8} concludes the paper. Given the growing prominence of expectiles as a class of risk measures, Appendix \ref{App:A} is devoted to analyzing the asymptotic behavior of DQ based on expectiles. All technical proofs are collected in Appendix \ref{App:B}.


\section{Definitions and preliminary conclusions}
\label{sec:2}
 Throughout this paper, we work on an atomless probability space $(\Omega,\;\mathcal{F},\;\mathbb{P})$. The assumption for atomless space is widely used in statistics and risk management; for further details, see \cite{D02} and Section A.3 of \cite{FS16}. 
 Let $\X$ be a convex cone of random variables on $(\Omega,\;\mathcal{F},\;\mathbb{P})$ where each $X \in \X$ represents the loss of a financial product. Let $L^1$ be  the set of all integrable variables on $(\Omega,\;\mathcal{F},\;\mathbb{P})$. Write $X\sim F$  if the random variable $X$ has the distribution function $F$  under $\mathbb{P}$, and  $X \laweq  Y$ if two random variables $X$ and $Y$ have the same distribution. 
	We always write $\mathbf X=(X_1,\dots,X_n)$ as a random vector and $S=\sum_{i=1}^n X_i$ as the sum of $\mathbf X$.  
	Further, denote by $[n]=\{1,\dots,n\}$.  Terms such as increasing or decreasing functions are in the non-strict sense. 
 A risk measure $\phi$ is a mapping from $\X$ to $\R$.  We will assume law-invariance for risk measures throughout the work.
 \begin{enumerate}
     \item[{[LI]}] Law-invariance: $\rho(X)=\rho(Y)$ for any $X, Y \in \X$ such that $X\laweq Y$. 
 \end{enumerate}
Law-invariance ensures that the value of a risk measure depends solely on the distribution of random variables, allowing us to apply empirical estimators for evaluation.
We collect some other properties of risk measures below.
\begin{enumerate}
\item[{[M]}] Monotonicity:  $\phi(X)\leq\phi(Y)$ for all $X, Y \in \X$ with  $X\leq Y$.
\item[{[TI]}] Translation invariance: $\phi(X+m)=\phi(X)+m$ for all $X \in \X$ and   $m\in\mathbb{R}$. 
\item[{[CV]}] Convexity: 
$\phi\left(\gamma X+(1-\gamma)Y\right)\leq\gamma \phi(X)+(1-\gamma )\phi(Y)$ for all  $X, Y \in \X$ and  $\gamma\in [0,1]$.
\end{enumerate} 
A   risk measure satisfies [M], [TI] and [CV] is said to be \emph{convex} (see \cite{FS16}). In risk management, many commonly used risk measures belong to a parametric family, such as  Value at Risk (VaR) and Expected Shortfall (ES),  two of the most popular classes of risk measures in banking and insurance. The definitions of VaR and ES are provided below.
\begin{enumerate}[(i)]
\item The Value at Risk (VaR) at level $\alpha \in [0,1)$ is defined by
$$
\VaR_\alpha(X)=\inf\{x\in \R: \p(X\le x) \ge 1-\alpha\}~~~X\in L^0.
$$
\item The Expected Shortfall (ES, also called CVaR, TVaR or AVaR) at level $\alpha \in (0,1)$ is defined by
$$
	\ES_{\alpha}(X) = \frac 1 \alpha \int_{0}^\alpha \VaR_\beta(X) \d \beta,~~~X\in L^1,
	$$
    and $\ES_0(X)=\esssup(X):=\VaR_0(X)$.
\end{enumerate}
We use the ``small-$\alpha$" definition here following the convention in \cite{HLW25}.

\cite{HLW25} defined the diversification quotient  (DQ)  based on a class of parameteric risk measures as following. 


\begin{definition}[Diversification Quotient]\label{def:DQ}
Let $  \rho =(\rho_{\alpha})_{\alpha \in  I}$ be a class of risk measures indexed by $\alpha\in I=(0,\overline \alpha) $ with $\overline \alpha\in(0,\infty]$ such that $\rho_\alpha$ is decreasing in $\alpha$. 
  For $\mathbf X \in \X^n$, the \emph{diversification quotient}  based on the class $ \rho$  at level $\alpha\in I$ is defined  by
 \begin{equation}\label{eq:DQ}
 {\rm DQ}^{\rho}_\alpha(\mathbf X)=\frac{\alpha^*}{\alpha}, \mbox{~~
where }
 \alpha^*= \inf\left\{\beta \in I :  \rho_{\beta} \left(\sum_{i=1}^n X_i\right) \le \sum_{i=1}^n \rho_{\alpha}(X_i) \right\}, 
 \end{equation} with the convention $\inf(\varnothing)=\overline \alpha$.
\end{definition}
It is clear that $(\VaR_\alpha)_{\alpha \in (0,1)}$ and $(\ES_\alpha)_{\alpha \in (0,1)}$  are decreasing classes of risk measures. When employing these risk measures, we typically prefer a level $\alpha$ close to zero, such as 0.01 or 0.025 in \cite{BASEL19}. By taking $\rho$ as VaR or ES class,  we have the two diversification indices $\DQVaR$ and $\DQES$ both of which have nice properties in terms of risk management; see \cite{HLW23, HLW25}  for more discussion.  

We first show that DQ based on a class of law-invariant risk measures is also law-invariant.
\begin{proposition}\label{prop:1}
Let $\rho$ be a class of law-invariant risk measures. Then, we have $\DQrho(\mathbf X)=\DQrho(\mathbf Y)$ for any $\mathbf X, \mathbf Y \in \X^n$ such that $\mathbf X \laweq\mathbf Y$.
\end{proposition}

For a class of law-invariant risk measures $\rho$, the value of DQ is determined by the joint distribution of the portfolio risk $\bX$. In this way, we can estimate the value of such DQ by the empirical distribution of the portfolio loss $\bX=(X_1, \dots, X_n)$.
Below, we simply write $\mathrm{DQ}_\alpha^\rho(F)=\mathrm{DQ}^\rho_\alpha(\bX)$ for $\bX \sim F$.

%

Let  $F$ be the joint distribution function of $\mathbf{X}$, $F_i$ be the  distribution of $X_i$ for $i\in [n]$, and $G$ be the distribution of $S=\sum_{i=1}^n X_i$. Suppose $\bX^{(1)}, \bX^{(2)}, \dots$ are i.i.d.~samples of $\bX$.  
For $N$ given data $\bX^{(1)}, \dots, \bX^{(N)}$ where $\bX^{(k)}=(X_1^{(k)}, \dots, X_n^{(k)})$ for $k \in [N]$, the empirical joint distribution is given by
 $$\widehat F^{N}(x_1, \ldots, x_n)=\frac{1}{N} \sum_{k=1}^N \id_{\{X^{(k)}_{1} \le x_1, \dots, X^{(k)}_{n} \leq x_n\}}, ~~~ (x_1,\dots, x_n) \in \mathbb{R}^n,$$
the empirical marginal distribution of $X_i$ is given by  $$\widehat F^{N}_i(x)=\frac{1}{N} \sum_{k=1}^N \id_{\{X^{(k)}_{i} \le x\}}, ~~~ x \in \mathbb{R},$$
and  the empirical distribution for $S$ is given by
$$\widehat G^{N}(x)=\frac{1}{N} \sum_{k=1}^N \id_{\{\sum_{i=1}^n X^{(k)}_{i} \le x\}}, ~~~ x \in \mathbb{R}.$$
By \citet[Theorem 1]{SW09},  if $\bX^{(1)}, \bX^{(2)}, \dots$ are i.i.d.~samples, the empirical joint distribution uniformly converges to the real joint distribution almost surely as $N \to \infty$; that is 
$$\sup_{\mathbf x \in \R^n}\vert\widehat F^{N}(\mathbf x)-F(\mathbf x)\vert \to 0 ~~~\p\mbox{-a.s.}$$
Similarly, the empirical  distributions $\widehat F_i^N$ for the marginal elements $X_i$, $i\in [n]$, and  $\widehat G^N$ for the sum $S$ are uniformly convergent to the real distributions $F_i$, $i\in [n]$,  and $G$.

For a law-invariant $\DQrho$  and a portfolio loss $\bX\in \X^n$, define the empirical estimator of $\DQrho(\bX)$ based on sample set $\{\bX^{(1)}, \dots, \bX^{(N)}\}$  by $$\wDQrho(N)(\omega)=\DQrho(\widehat F (\omega)), ~~~ \omega \in \Omega,$$ where $\widehat F (\omega)$ is the realization of $\wF$ on $\omega$.
 Note that $\widehat\DQrho(N)$ is a real-value random variable determined by sample set  $\{\bX^{(1)}, \dots, \bX^{(N)}\}$.
We will mainly focus on the statistic property of $\DQVaR$ and $\DQES$. 

\section{Consistency property of DQ estimator}\label{sec:3}
 
 To prepare the discussion of  consistency for DQ's empirical estimators, we first present some convergence properties for DQ with respect to distribution.
Let $\{\bX^{(k)}\}_{k \in \N}$ be a sequence of random vectors such that $\bX^{(k)} \lawto \bX$ as $k \to  \infty$. 
A classic result for law-invariant convex risk measures on $L^1$ space, including ES, is the continuity with respect to the $L^1$-norm (See \citet[Remark 1.1]{LCLW20}).  Together with the uniform integrability in the following assumption, we can obtain the continuity of law-invariant convex risk measures with respect to distribution; see \cite{HW24} for more discussion on uniform integrability and risk measures. 
\begin{assumption}\label{ass:UI}
All $\{X^{(k)}_i\}_{k \in \N}$, $i\in [n]$, and $\{\sum_{i=1}^n X^{(k)}_i\}_{k \in \N}$ are uniformly integrable.
\end{assumption}
Note that  VaR is not a convex risk measure. Weak convergence of VaR requires the continuity of the quantile function for the limit distribution. We make the following assumption to discuss the continuity for $\DQVaR$.
\begin{assumption}\label{ass:quantile}
The quantile function of  $X_i$ is continuous at $1-\alpha$ for each $i\in [n]$,  and $\sum_{i=1}^n X_i$ have continuous quantile function.
\end{assumption}

 Let $\widetilde \alpha=\sup\{\beta\in I: \rho_\beta(\sum_{i=1}^n X_i)\ge \sum_{i=1}^n \rho_\alpha(X_i)\}$ with the convention $\sup(\varnothing)=0$.
With the above assumptions, we can obtain the convergence properties of DQ based on a class of law-invariant convex risk measures or the VaR class.
\begin{theorem}\label{thm:1}
Let $\{\bX^{(k)}\}_{k \in \N}\subset \X^n$, $\bX \in \X^n$ such that $\bX^{(k)} \lawto \bX$ as $k \to \infty$. If either
\begin{enumerate}[(a)]
\item $\rho=(\rho_{\alpha})_{\alpha\in I}$ is a class of law-invariant convex risk measures and Assumption \ref{ass:UI} holds, or
\item $\rho=(\VaR_\alpha)_{\alpha \in I}$ with $I =(0,1)$ and Assumption \ref{ass:quantile} holds,
\end{enumerate}
then, for all $\alpha \in I$,
$$
\liminf_{k\to \infty}\DQrho(\mathbf X^{(k)})\ge \DQrho(\mathbf X),~~~\mbox{and }~~~\limsup_{k\to\infty}\DQrho(\bX^{(k)})\le \frac{\widetilde{\alpha}
}{\alpha}.$$
If  further
$\beta \mapsto \rho_\beta(\sum_{i=1}^n X_i)$ is a strictly decreasing function, then $\lim_{k \to \infty }\DQrho(\bX^{(k)}) = \DQrho(\bX)$ for all $\alpha \in I$.  
\end{theorem}




A direct proposition of Theorem \ref{thm:1} is the consistency of the empirical estimator for DQ, which we state below.

 \begin{proposition}
     \label{prop:consistency}
       \begin{enumerate}[(a)]
\item $\rho=(\rho_{\alpha})_{\alpha\in I}$ is a class of law-invariant convex risk measures  and Assumption \ref{ass:UI} holds, or
\item $\rho=(\VaR_\alpha)_{\alpha \in I}$ with $I=(0,1)$ and  Assumption \ref{ass:quantile} holds,
\end{enumerate}
       then
$$
\liminf_{N\to \infty}\wDQrho(N)\ge \DQrho(\mathbf X),~~~\mbox{and }~~~\limsup_{N\to\infty}\wDQrho(N)\le \frac{\widetilde{\alpha}
}{\alpha} ~~\p\mbox{-a.s.}$$
 for any $\alpha \in I$.  If further $\beta \mapsto \rho_\beta(S)$ is a strictly decreasing function, then  $\DQrho(N)$ is a strongly consistent empirical estimator for $\DQrho(\bX)$; that is, $\lim_{N \to \infty} \wDQrho(N) =\DQrho(\mathbf X)$ $\p$-a.s.
 \end{proposition}
It  the total loss $S$ is not degenerate, then the function   $\beta \mapsto \ES_\beta(S)$  is  strictly decreasing. 
By Proposition \ref{prop:consistency},  under  Assumption  \ref{ass:UI},  the empirical estimators $ \widehat{\DQES}(N) $  is strongly consistent. If $S$ is  constant, then $\wDQVaR(N)=0$ a.s. Hence,  if the portfolio $\bX$ satisfies Assumption \ref{ass:quantile} and the total loss $S$ has continuous and strictly decreasing quantile function, $\wDQVaR(N)$ is also a strongly consistent estimator.

  In the next section, we further explore the asymptotic normality of the empirical estimators for  DQs based on VaR and ES.

\section{Asymptotic normality on i.i.d.~data}\label{sec:4}

We first collect the alternative formulas  for $\DQVaR$ and  $\DQES$  from Theorem 4 of \cite{HLW25}.

\begin{proposition}\label{prop:3}
 For a given $\alpha\in(0,1)$, $\mathrm{DQ}^{\mathrm{VaR}}_{\alpha}$ have the alternative formula
\begin{equation}\label{eq:DQ_VaR}\mathrm{DQ}^{\mathrm{VaR}}_{\alpha}(\textbf{X})=\frac{1}{\alpha}\mathbb{P}\left(S>\sum_{i=1}^{n}\rm{VaR}_{\alpha}(X_i)\right).\end{equation}
For a given $\alpha \in (0,1)$, if $\p(\sum_{i=1}^n  X_i>\sum_{i=1}^n \ES_\alpha(X_i) )>0$, then 
$${\rm DQ}^{\ES}_\alpha(\mathbf X)= \frac{1}{\alpha}\min_{r\in (0,\infty)} \E\left[\left(r \sum_{i=1}^n (X_i-\ES_\alpha(X_i))+1\right)_+\right],$$
and otherwise ${\rm DQ}^{\ES}_\alpha(\mathbf X)=0.$
\end{proposition}

Proposition \ref{prop:3} provides convenient formulas for computing the empirical estimators of $\DQVaR$ and $\DQES$.
Let $\bX^{(1)}, \dots, \bX^{(N)}$ be i.i.d.~samples of $\bX=(X_1, \dots, X_n)$ where $\bX^{(k)}=(X_1^{(k)}, \dots, X_n^{(k)})$ for $k \in [N]$. For $i\in [n]$, let $\widehat x^{\VaR_\alpha}_i$ and $\widehat x^{\ES_\alpha}_i$  be the empirical estimators for $\VaR_\alpha(X_i)$ and $\ES_\alpha(X_i)$.
For a given $\alpha\in(0,1)$, the empirical estimator for $\mathrm{DQ}^{\mathrm{VaR}}_{\alpha}$ is given by 
\begin{equation}
\widehat{\DQVaR}(N)=\frac{1}{N\alpha}{\sum_{k=1}^N \id_{\{\sum_{i=1}^n X^{(k)}_i> \sum_{i=1}^n\widehat x^{\VaR_\alpha}_i\}}}.
\end{equation}
For a given $\alpha \in (0,1)$, if there exists $k \in [N]$ such that $\sum_{i=1}^n  X^{(k)}_i>\sum_{i=1}^n \widehat x^{\ES_\alpha}_i$, then  the empirical estimator for $\DQES$ is given by 
\begin{equation}
\widehat{\DQES}(N)= \frac{1}{N\alpha}\min_{r\in (0,\infty)} \sum_{k=1}^N\left[\left(r \sum_{i=1}^n (X^{(k)}_i-\widehat x^{\ES_\alpha}_i)+1\right)_+\right],\end{equation}
and otherwise ${\rm DQ}^{\ES}_\alpha(\mathbf X)=0.$

%
%
%
%
%

We start with the asymptotic normality of $\wDQVaR(N)$. 
For technical reasons, it is necessary to impose certain properties on the distribution function of $\mathbf{X}$.
\begin{assumption}\label{ass:DQVaR}
     For each $i \in[n]$,  the distribution function $F_i$ has a  positive  density  $f_i$  at  $F_i^{-1}(1-\alpha)$. The distribution function $G$  has a  density   $g$ that  is bounded  and continuous at $\sum_{i=1}^{n}F_i^{-1}(1-\alpha)$. 
\end{assumption}

Such regularity assumptions are standard in the literature on quantile inference and risk measures; see, e.g. \cite{B66} and \cite{K67}. Note that Assumption \ref{ass:DQVaR} is stronger than Assumption \ref{ass:quantile},  which  allows us to apply Bahadur representation (see Lemma \ref{lem:Bahadur} in Appendix \ref{app:2}) on the marginal empirical quantile functions $(\wF^N)^{-1}_i(\alpha)$, $i\in [n]$,  and the mean-value theorem on $G$.

\begin{theorem}[Asymptotic normality of $\mathrm{DQ}_\alpha^{\VaR}$]\label{thm:2}
 Suppose that $\bX^{(1)}, \bX^{(2)} \dots \in \X^n$  are  i.i.d.~samples of $\bX=(X_1, \dots, X_n)$, $\alpha \in (0,1)$ and  Assumption \ref{ass:DQVaR} holds.  Let $t_i=F_i^{-1}(1-\alpha)$ for $i\in [n]$ and $ t_{n+1}=\sum_{i=1}^{n}t_i$. 
 As $N \rightarrow \infty$, we have
$$
\sqrt{N}\left(\wDQVaR(N)-\mathrm{DQ}^{\VaR}_\alpha(\mathbf X)\right)\lawto \rm{N}(0,\sigma_{\VaR}^2), 
$$
with $\sigma_{\VaR}^2=\mathbf A_\VaR^\top\Sigma_\VaR\mathbf A_\VaR$ where
	$$
 	\mathbf A_\VaR=\left(\dfrac{g(t_{n+1})}{\alpha f_1(t_1)},\dots ,\dfrac{g(t_{n+1})}{\alpha f_n(t_n)},-\frac{1}{\alpha}\right),
 	$$
and $\Sigma_\VaR$ is the covariance matrix of random vector $(\id_{\{X_1 \le t_1\}}, \dots, \id_{\{X_n \le t_n\}}, \id_{\{S\le t_{n+1}\}})$.
\end{theorem}	
\begin{remark}
Let $X_{n+1}=S$. A simple calculation  yields that $\Sigma_\VaR=(\sigma^2_{ij})_{(n+1)\times (n+1)}$ where 
\begin{align*} 
\sigma^2_{ij}&=\cov\left(\id_{\{X_i\le t_i\}},\id_{\{X_j \le t_j\}}\right)=\p(X_i\leq t_i,X_j\leq t_j)-\p(X_i\leq t_i)\p(X_j\leq t_j)~~ \mbox{for}  ~~i,j \in [n+1].
\end{align*}
 The covariance matrix is uniquely determined by the joint distribution of $\bX$.
\end{remark}

The following two assumptions are required for the asymptotic normality of $\wDQES$.
 \begin{assumption}\label{ass:unique}
$\sum_{i=1}^{n}\ES_{\alpha}(X_i)< \esssup (S)$. 
\end{assumption}
Assumption \ref{ass:unique} implies that $\DQES(\bX)>0$ by \citet[Theorem 3]{HLW23}. Together with the fact that $\beta \mapsto \ES_\beta(S)$ is a continuous function, $\alpha^*$ satisfies $\ES_{\alpha^*}(S)=\sum_{i=1}^n \ES_\alpha(X_i)$.


\begin{assumption}\label{ass:DQES}
For each $i\in [n]$, the distribution function $F_i$  has a positive density  $f_i$ on $[F_i^{-1}(1-\alpha), \infty)$ and  $\mathbb{E}[|X_i|^{2+\epsilon_i}]<\infty$ for some $\epsilon_i>0$. The distribution fucntion $G$ has a density  $g$ which is positive on $[G^{-1}(1-\alpha^*),\infty)$ and  continuous at $G^{-1}(1-\alpha^*)$, and   $\mathbb{E}[|S|^{2+\epsilon}]<\infty$ for some $\epsilon>0$.  
\end{assumption}
We need slightly more than the second moment for each component $X_i$ and the sum $S$, which is a standard regularity condition required for the convergences of the weighted sum of empirical quantile process; see \cite{SY96} for more discussion about the regularity condition. 


The following theorem characterizes the asymptotic normality of $\DQES$. 
\begin{theorem}[Asymptotic normality of $\mathrm{DQ}^{\mathrm{ES}}_\alpha$]\label{thm:3} 
 Suppose that  $\bX^{(1)}, \bX^{(2)} \dots \in \X^n$ are i.i.d.~samples of $\bX=(X_1, \dots, X_n)$, $\alpha \in (0,1)$ and  Assumptions \ref{ass:unique} and \ref{ass:DQES} hold. 
  Let $\alpha^*=\alpha \DQES(\mathbf X)$, $t_i=F_i^{-1}(1-\alpha)$ for $i\in [n]$ and $s=G^{-1}(1-\alpha^*)$.  As $N \to \infty$, we have 
\begin{equation*}\sqrt{N}\left(\wDQES(N)-\DQES(\bX)\right)\lawto \mathrm{N}(0, \sigma_{\ES}^2),\end{equation*}
with $\sigma_\ES^2=\mathbf A_\ES^\top \Sigma_\ES \mathbf A_\ES/c^2$ where $\mathbf A_\ES \in \R^{n+1}$ such that
$\mathbf A_\ES=( 1/\alpha, \dots,1/\alpha, -1/\alpha^*)$, $\Sigma_\ES$ is the covariance matrix of random vector $( (X_1-t_1)_+, \dots, (X_n-t_n)_+, (S-s)_+)$ and $c=\left( \VaR_{\alpha^*}(S)-\ES_{\alpha^*}(S)\right)/\DQES(\bX)$.
\end{theorem}

\begin{remark} 
Note that the proof for the asymptotic distribution of $\mathrm{DQ}^{\ES}_\alpha$  is similar to the asymptotic normality for the empirical estimator of the probability equivalence level of VaR and ES (PELVE) proved in \cite{LW23}. For $\epsilon \in (0,1)$,  PELVE is defined as 
$$
\Pi_{\varepsilon}(X) = \inf \left\{ c \in [1, 1 / \varepsilon] : \operatorname{ES}_{1 - c \varepsilon}(X) \leq \operatorname{VaR}_{1 - \varepsilon}(X) \right\},~~ X \in L^1,
$$ 
which has similar definition as DQ in Definition \ref{def:DQ}.  It is worth noting that PELVE is a one-dimensional quantity, while DQ is defined on a  vector, which makes the analysis more involved. The same technology fails for the empirical estimator of $\DQVaR$ since the empirical quantile function is not differentiable. 
\end{remark}

\section{DQ based on dependent data}\label{sec:5}
In this section, we further extend the results in Sections \ref{sec:3} and \ref{sec:4} to dependent data where the sequence $\{X^{(k)}\}_{k \ge 1}$ is $\alpha$-mixing. Below, we present  the definition of $\alpha$-mixing.

\begin{definition}[$\alpha$-mixing]
A sequence $\{\bX^{(k)}=(X_1^{(k)}, \dots, X_n^{(k)})\}_{k \ge 1}$ is $\alpha$-mixing if its mixing coefficient 
$$\alpha( m):=\sup\{\vert \p(A \cap B) -\p(A)\p(B)\vert: A \in \mathcal{F}_{1}^k, B \in \mathcal{F}_{k+m}^\infty, 1\le k<\infty \}\to 0$$
as $m \to \infty$, where $\mathcal{F}_a^b$ denotes the $\sigma$-field generated by $\{\bX^{(k)}\}_{a\le k\le b}$.
\end{definition}
We slightly abuse the notation $\alpha$ here. The mixing coefficient $\alpha(\cdot)$ is a function of integer number $m$ instead of probability level $\alpha$ for $\mathrm{DQ}_{\alpha}^\rho$. The i.i.d.~sample sequence is a special case of $\alpha$-mixing sequence with $\alpha(m)=0$ for all $m \in \mathbb N$. 

For a function  $f: \R^n \to \R$, it is clear that the mixing coefficient $\alpha_{\{f(\bX^{(k)})\}}(m)\le \alpha_{\{\bX^{(k)}\}}( m)$. Hence, if  $\{\bX^{(k)}\}_{k \ge 1}$ is $\alpha$-mixing, $\{f(\bX^{(k)})\}_{k \ge 1}$ is $\alpha$-mixing.  By taking $f$ as sum or $i$-th element of vector,  we have $\{\sum_{i=1}^n X_i^{(k)}\}_{k \ge 1}$ and $\{X_i^{(k)}\}_{k \ge 1}$ are $\alpha$-mixing for all $i \in [n]$.

To the best of our knowledge, the uniformly convergence of empirical joint distributions based on $\alpha$-mixing data is not available in the existing literature. Nevertheless, the empirical distribution for a real value random variable is uniformly convergent with  $\alpha$-mixing samples under some certain condition as shown in \citet[Corollary 2.1]{CR92}. 
Hence, we need an additional assumption on the mixing coefficient for the consistency of $\wDQrho(N)$ under $\alpha$-mixing data.

\begin{assumption}\label{ass:alpha-emp}
The mixing coefficient $\alpha(m)$ for $\{\bX^{(k)}\}_{k \ge 1}$ satisfies
$$\sum_{m=1}^ \infty m^{-1} (\log m) (\log \log m)^{1+\delta} \alpha(m) < \infty~~~\mbox{for some}~~~\delta>0.$$
\end{assumption}

Under Assumption \ref{ass:alpha-emp}, we can extend Theorem  \ref{thm:1} to $\alpha$-mixing data.
 \begin{proposition}\label{prop:alpha-con}
       Suppose $\{ \bX^{(k)}\}_{k \ge 1}$ is a stationary  $\alpha$-mixing  sequence with Assumption \ref{ass:alpha-emp}. If either
       \begin{enumerate}[(a)]
\item $\rho=(\rho_{\alpha})_{\alpha\in I}$ is a class of law-invariant convex risk measures, or
\item $\rho=(\VaR_\alpha)_{\alpha \in I}$ with $I=(0,1)$ and  Assumption \ref{ass:quantile} holds,
\end{enumerate}
       then
$$
\liminf_{N\to \infty}\wDQrho(N)\ge \DQrho(\mathbf X),~~~\mbox{and }~~~\limsup_{N\to\infty}\wDQrho(N)\le \frac{\widetilde{\alpha}
}{\alpha} ~~\p\mbox{-a.s.}$$
 for any $\alpha \in I$.  If further $\beta \mapsto \rho_\beta(S)$ is a strictly decreasing function, then  $\DQrho(N)$ is a strongly consistent empirical estimator for $\DQrho(\bX)$; that is, $\lim_{N \to \infty} \wDQrho(N) =\DQrho(\mathbf X)$ $\p$-a.s.
 \end{proposition}
 Assumption \ref{ass:alpha-emp} provides uniformly convergence of the empirical distributions for marginal $X_i$, $i\in [n]$, and sum $S$ with stationary $\alpha$-mixing data. Since the convergence of $\DQrho$ only relies on the weak convergence of $X_i$, $i\in [n]$ and $S$, we can simply extend the proof of Theorem  \ref{thm:1} to show Proposition \ref{prop:alpha-con}.
 Note that Assumption \ref{ass:UI} is not required in Proposition \ref{prop:alpha-con} (a) since stationary sequence satisfies Assumption \ref{ass:UI}. 

To show the asymptotic normality of $\wDQVaR$ and $\wDQES$, we need stronger  assumptions  to extend the strong convergence of empirical process $E^N$ defined in Lemma \ref{lem:1} (Lemma \ref{lem:emp-alpha} in Appendix \ref{app:depend}), Bahadur's representation (Lemma \ref{lem:quan-alpha} in Appendix \ref{app:depend}) and central limit theorem to $\alpha$-mixing sequences (Lemma \ref{lem:CLT-alpha} in Appendix \ref{app:depend}).

\begin{assumption}\label{ass:alpha}
$\{\bX^{(k)}\}_{k \ge 1}$  is an $\alpha$-mixing  stationary sequence with $\alpha(m)=\O(m^{-5-\epsilon})$ for some $0<\epsilon \le 1/4$. Moreover, $\p(S^{(k_1)}=S^{(k_2)})=0$ and   $\p(X^{(k_1)}_i=X^{(k_2)}_i)=0$ for any $k_1\neq k_2$ and $i\in [n]$. 
\end{assumption}

\begin{assumption}\label{ass:DQVaR-alpha}
     For each $i \in[n]$,   the distribution function $F_i$ has a positive density  $f_i$ where $f_i(x)$ and $f_i'(x)$ are bounded in some neighborhood of $F_i^{-1}(p)$. The distribution function  $G$ has a  density   $g$ which is bounded  and continuous at $\sum_{i=1}^{n}F_i^{-1}(1-p)$. 
\end{assumption}

The requirement for mixing coffecifient $\alpha(m)$ in Assumption \ref{ass:alpha} is stronger than Assumption \ref{ass:alpha-emp} and requirement for distributions $F_i$, $i \in [n]$, and $G$ in Assumption \ref{ass:DQVaR-alpha} is stronger than Assumption \ref{ass:quantile}. Hence, with Assumptions \ref{ass:alpha} and \ref{ass:DQVaR-alpha}, $\wDQVaR$ is a strongly consistent estimator by Proposition \ref{prop:alpha-con}.

For two random vectors $\mathbf X=(X_1, \dots, X_n)$ and $\mathbf Y=(Y_1, \dots, Y_n)$, their covariance matrix is defined as 
\begin{align*}
\cov(\mathbf X, \mathbf Y)=\left(\cov(X_i, Y_j)\right)_{i,j \in [n]}.
\end{align*}
We present the asymptotic normality of $\mathrm{DQ}_\alpha^{\VaR}$ for $\alpha$-mixing sequences with Assumptions  \ref{ass:alpha} and \ref{ass:DQVaR-alpha}.
\begin{theorem}[Asymptotic normality of $\mathrm{DQ}_\alpha^{\VaR}$ for $\alpha$-mixing sequences]\label{thm:alphaDQVaR}
 Suppose that $\{\bX^{(k)}\}_{k \ge 1}$ satisfies Assumptions \ref{ass:alpha} and \ref{ass:DQVaR-alpha}, and  $p \in (0,1)$.  Let $t_i=F_i^{-1}(1-p)$ for $i\in [n]$ and $ t_{n+1}=\sum_{i=1}^{n}t_i$. 
 As $N \rightarrow \infty$, we have
$$
\sqrt{N}\left(\wDQVaR(N)-\mathrm{DQ}^{\VaR}_\alpha(\mathbf X)\right)\lawto \rm{N}(0,\sigma_{\VaR}^2), 
$$
with $\sigma_{\VaR}^2=\mathbf A_\VaR^\top\Sigma^\alpha_\VaR\mathbf A_\VaR$ where
	$$
 	\mathbf A_\VaR=\left(\dfrac{g(t_{n+1})}{\alpha f_1(t_1)},\dots ,\dfrac{g(t_{n+1})}{\alpha f_n(t_n)},-\frac{1}{\alpha}\right),
 	$$
and $\Sigma^\alpha_{\VaR}=\var(\mathbf I^{(1)})+2\sum_{k=2}^\infty\cov(\mathbf I^{(1)},\mathbf I^\k)$
and  $\mathbf I^{(k)}=(\id_{\{X_1^{(k)}\le t_1\}}, \dots, \id_{\{X_n^{(k)}\le t_n\}})$.\label{thm:4}
\end{theorem}	
 Next, we discuss the asymptotic normality of $\mathrm{DQ}_\alpha^{\ES}$ for $\alpha$-mixing sequences based on the following assumption.

 \begin{assumption}\label{ass:DQES-alpha}
For each $i\in [n]$, the distribution function $F_i$  has a positive continuous density $f_i(x)$ where $f_i(x)$ and $f'_i(x)$ are bounded in some neighborhood of $F_i^{-1}(p)$ and  $\mathbb{E}[|X_i|^{2+\epsilon_i}]<\infty$ for some $\epsilon_i>1/2$. The distribution  function $G$ has a positive continuous density $g(x)$ where $g(x)$ and $g'(x)$ are bounded in some neighborhood of $g^{-1}(p)$ and   $\mathbb{E}[|S|^{2+\epsilon}]<\infty$ for some $\epsilon>1/2$.  
\end{assumption}

\begin{theorem}[Asymptotic normality of $\mathrm{DQ}_\alpha^{\ES}$ for $\alpha$-mixing sequences]\label{thm:alphaDQES} 
 Suppose that  $\{\bX^{(k)}\}_{k \ge 1}$ satisfies Assumptions  \ref{ass:unique}, \ref{ass:alpha} and \ref{ass:DQES-alpha}, and $p \in (0,1)$.
  Let $\alpha^*=\alpha \DQES(\mathbf X)$, $t_i=F_i^{-1}(1-\alpha)$ for $i\in [n]$ and $s=G^{-1}(1-\alpha^*)$.  As $N \to \infty$, we have 
\begin{equation*}\label{eq:sig_ES}\sqrt{N}\left(\wDQES(N)-\DQES(\bX)\right)\lawto \mathrm{N}(0, \sigma_{\ES}^2)\end{equation*}
with $\sigma_\ES^2=\mathbf A_\ES^\top \Sigma^\alpha_\ES \mathbf A_\ES/c^2$ where  $c=\left( \VaR_{\alpha^*}(S)-\ES_{\alpha^*}(S)\right)/\DQES(\bX)$, 
$\mathbf A_\ES=( 1/\alpha, \dots,1/\alpha, -1/\alpha^*)$ and  $\Sigma^\alpha_\ES=\var( \mathbf Y^{(1)})+2\sum_{k=2}^\infty \cov( \mathbf Y^{(1)}, \mathbf Y^{(k)})$,  $\mathbf Y^{(k)}=( (X^{(k)}_1-t_1)_+, \dots, (X^{(k)}_n-t_n)_+, (S^{(k)}-s)_+)$ for $k \ge 1$.\label{thm:5}
\end{theorem}
In the case that $\{\bX^{(k)}\}_{k \ge 1}$ is a sequence of i.i.d.~random vectors, the coveriance terms in  Theorems \ref{thm:alphaDQVaR} and \ref{thm:alphaDQES} disappear; hence, Theorems \ref{thm:alphaDQVaR} and \ref{thm:alphaDQES} degenerate to Theorems \ref{thm:2} and \ref{thm:3}, respectively. 
\section{Comparison of  DQ and DR}\label{sec:6}

For a given risk measure $\phi: \X \to \R$, a well-studied diversification index, called diversification ratio (DR), is defined by 
$$
\mathrm{DR}^\phi(\bX)=\frac{\phi(S)}{\sum_{i=1}^n \phi(X_i)}, ~~~\bX \in \X^n.$$
Let $\widehat{\phi}^N(X_i)$ be the empirical estimators for $\phi(X_i)$ for $i \in [n]$ and  $\widehat{\phi}^N(S)$ be the empirical estimators for $\phi(S)$. Assume we have the asymptotic normality for $(\widehat{\phi}^N(X_i), \dots, \widehat{\phi}^N(X_i), \widehat{\phi}^N(S))$; that is,
$$(\widehat{\phi}^N(X_1)-\phi(X_1), \dots, \widehat{\phi}^N(X_n)-\phi(X_n), \widehat{\phi}^N(S)-\phi(S)) \lawto \mathrm{N}(0, \Sigma_\phi).$$
By Delta method,  the empirical estimator for DR satisfies
$$\widehat{\mathrm{DR}^\phi}(N)-\mathrm{DR}^\phi(\bX)\lawto \mathrm{N}(0, \sigma^2_\phi),$$
with $\sigma^2_\phi=\mathbf R_\phi^\top \Sigma_\phi \mathbf R_\phi$
where $\mathbf{R}_\phi \in \R^{n+1} $ such that $$\mathbf{R}_\phi=\left(-\frac{\phi(S)}{(\sum_{i=1}^n  \phi(X_i))^2}, \dots, -\frac{\phi(S)}{(\sum_{i=1}^n  \phi(X_i))^2}, \frac{1}{\sum_{i=1}^n \phi(X_i)}\right).$$
An obvious drawback for DR empirical estimator is that $\sigma^2_\phi \to \infty$ if $\sum_{i=1}^n \phi(X_i) \to 0$. As a result, the empirical estimator for DR based on a cash-additive risk measure fails if we consider a centralized portfolio loss $\mathbf Y=(X_1-\phi(X_1), \dots, X_n -\phi(X_n))$. That is, we should be very careful in the location of the portfolio when we estimate DR. 
This disadvantage is avoided by DQ based on a class of cash-additive risk measures since DQ is location invariance. More precisely, if $\rho=\VaR$ or $\ES$, then  $\mathrm{DQ}_\alpha^\rho(\mathbf{X}+\mathbf{c})=\mathrm{DQ}_\alpha^\rho(\mathbf{X})$  for all $\mathbf{c}=(c_1,\dots,c_n) \in \R^n$ and all $ \mathbf X\in \X^n$.  Moreover, as we can see from Theorems  \ref{thm:2} and \ref{thm:3}, $\sigma^2_\VaR$ and $\sigma^2_\ES$   are finite under certain conditions.

Below, we present the explicit asymptotic distributions for $\widehat{\mathrm{DR}^\VaR}(N)$ and $\widehat{\mathrm{DR}^\ES}(N)$.

\begin{proposition}[Asymptotic normality of $\mathrm{DR}$]\label{prop:5}
Suppose that    $\bX^{(1)}, \bX^{(2)}, \dots\in\X^n$ are i.i.d.~samples of $\bX=(X_1, \dots, X_n)$.  Let  $s=G^{-1}(1-\alpha)$, $t_i=F_i^{-1}(1-\alpha)$ for $i\in [n]$ and $ t_{n+1}=\sum_{i=1}^{n}t_i$.  The following statements hold. 
\begin{itemize}
    \item[(i)] Suppose   $\alpha\in (0,1)$ and that Assumption \ref{ass:DQVaR} holds. As $N\to\infty$, we have  $$\sqrt{N}\left(\widehat{\mathrm{DR}^{\VaR_{\alpha}}}(N)-\mathrm{DR}^{\VaR_{\alpha}}(\mathbf{X})\right)\stackrel{d}{\longrightarrow} \mathrm{N}(0, \widetilde \sigma^2_{\VaR_{\alpha}}),$$
   with $\widetilde \sigma^2_{\VaR_{\alpha}}=\mathbf R_{\VaR}^\top\Sigma_{\VaR} \mathbf R_{\VaR}$, where  where  $$\mathbf{R}_{\VaR}=\left(\frac{s}{f_1(t_1)t_{n+1}^2}, \dots, \frac{s}{f_n(t_n) t_{n+1}^2}, \frac{-1}{g(s)t_{n+1}}\right),$$  and $\Sigma_{\VaR}$ is the covariance matrix of random vector $(\id_{\{X_1 \le t_1\}}, \dots, \id_{\{X_n \le t_n\}}, \id_{\{S\le s\}}).$
\item[(ii)]Suppose   $\alpha\in(0,1)$ and Assumption \ref{ass:DQES} holds. As $N\to\infty$, we have 
$$\sqrt{N}\left(\widehat{\mathrm{DR}^{\ES_{\alpha}}}(N)-\mathrm{DR}^{\ES_{\alpha}}(\mathbf{X})\right)\lawto \mathrm{N}(0,\widetilde\sigma^2_{\ES}),$$ 
with $\widetilde\sigma^2_{\ES}=\mathbf R^\top_{\ES}\Sigma_{\ES} \mathbf R_{\ES}$, where
 $$\mathbf{R}_{\ES}=\frac{1}{\alpha}\left(-\frac{\ES_\alpha(S)}{\left(\sum_{i=1}^n\ES_\alpha(X_i)\right)^2}, \dots, -\frac{\ES_\alpha(S)}{\left(\sum_{i=1}^n\ES_\alpha(X_i)\right)^2},  \frac{1}{\sum_{i=1}^n\ES_\alpha(X_i)}\right),$$
and  $\Sigma_{\ES}$ 
is  the covariance matrix of random vector  $( (X_1-t_1)_+, \dots, (X_n-t_n)_+,(S-s)_+)$. 
\end{itemize}  
\end{proposition}

\section{Simulation study}\label{sec:7}

This section verifies the asymptotic normality for the empirical estimator under the assumption that the portfolio follows an elliptical distribution. Elliptical distributions represent a fundamental class of multivariate distributions, characterized by their symmetry and flexible dependence structures, see \cite{MFE15}. Specifically, a random vector  is elliptically distributed if its characteristic function can be expressed as:
	$$
	\begin{aligned}
		\psi(\mathbf{t}) =\mathbb{E}\left[\exp \left(\texttt{i} \mathbf{t}^\top \mathbf{X}\right)\right] & 
		=\exp \left(\texttt{i} \mathbf{t}^\top \boldsymbol{\mu}\right) \tau\left(\mathbf{t}^\top \Sigma \mathbf{t}\right),
	\end{aligned}
	$$
where  $\boldsymbol{\mu}\in \mathbb{R}^{n}$  represents the location parameter,   $ \Sigma\in \R^{n\times n}$ is a positive semi-definite dispersion matrix, and  $\tau: \mathbb{R}_{+} \rightarrow \mathbb{R}$ is known as the characteristic generator.
	Such distributions are denoted by $ \mathrm{E}_{n}(\boldsymbol{\mu}, \Sigma, \tau).
	$ 
	Prominent examples of elliptical distributions include the normal and 
t-distributions, both of which are widely applied in risk management and financial modeling.
 
 For  a positive semi-definite matrix $\Sigma$,
	we write $\Sigma=(\sigma_{ij})_{n\times n}$, $\sigma_i^2=\sigma_{ii}$, and $\boldsymbol \sigma=(\sigma_1,\dots,\sigma_n)$, and
	define the constant
	\begin{equation*}\label{eq:k}k_\Sigma= \frac {\sum_{i=1}^n 
			\left(\mathbf{e}^\top_i \Sigma \mathbf{e}_i \right)^{1/2}} {\left( \mathbf{1}^\top \Sigma \mathbf{1}\right)^{1/2}  } 
	=\frac{\sum_{i=1}^n\sigma_{i} }{ \left(\sum_{i, j}^n \sigma_{ij}\right)^{1/2} }
	\in [1,\infty),\end{equation*}
where   $\mathbf{1}=(1,\dots,1)\in\R^n$ and  $ \mathbf e_{1},\dots, \mathbf e_{n}$ are the column vectors of the $n\times n$ identity matrix $I_n$. 

The representations of DQs based on VaR and ES for elliptical distributions can be explicitly derived, as shown in the following proposition. The proofs can be found in Proposition 2 of \cite{HLW23}. We define the \emph{superquantile transform} of a distribution $F$ with finite mean
	is  a distribution $\overline F$ with quantile function $\alpha\mapsto \ES_{1-\alpha}(X)$ for $p\in(0,1)$, where $X\sim F$.
\begin{proposition}
	\label{prop:comp_Dvar}
	Suppose that $\mathbf X \sim  \mathrm{E}_{n}(\boldsymbol{\mu}, \Sigma, \tau)$.  We have, for $\alpha \in (0,1)$,
	$${\rm DQ}_\alpha^{\VaR}(\mathbf X)=
		\frac{1-F  (k_\Sigma \VaR_{\alpha}(Y)  )}{\alpha},
~~
		{\rm DQ}_\alpha^{\ES}(\mathbf X)=
		\frac{1- \overline F (k_\Sigma \ES_{\alpha}(Y)  )}{\alpha},	
	$$
		where $ Y \sim \mathrm{E}_1(0,1,\tau)$ with  distribution function  $F$, and  $ \overline F$ is the superquantile transform of  $F$.
\end{proposition}

Next, we take a close look at the asymptotic normality of DQ  for the two most important elliptical distributions used in finance
and insurance:  multivariate normal distribution and the multivariate t-distribution.  The explicit formulas  provided in Proposition \ref{prop:comp_Dvar} are used to calculate the true values of the DQs. For numerical illustrations, we consider     an equicorrelated model parameterized by $r\in [0,1]$ and $n\in \N$,
	\begin{equation}\label{eq:matrix}
	\Sigma=(\sigma_{ij})_{n\times n},~~~\mbox{ where $\sigma_{ii}=1$ and $\sigma_{ij}=r$ for $i\ne j$}.
	\end{equation}
We take two  models $\mathbf{X} \sim\mathrm N(\boldsymbol{\mu},\Sigma)$ and $ \mathbf{Y}\sim \mathrm  t(\nu,\boldsymbol{\mu},\Sigma)$.  Note that the location $\boldsymbol{\mu}$ does not matter in computing DQ, and we can simply take $\boldsymbol \mu=\mathbf 0$. The default parameters are set as $r=0.3$, $n=5$, $\nu=3$  and $\alpha =0.1$ if not explained otherwise. We generate $N=5000$ samples once and repeat the simulation 2000 times.  
\begin{figure}[htb!]
\centering
 \includegraphics[width=16cm]{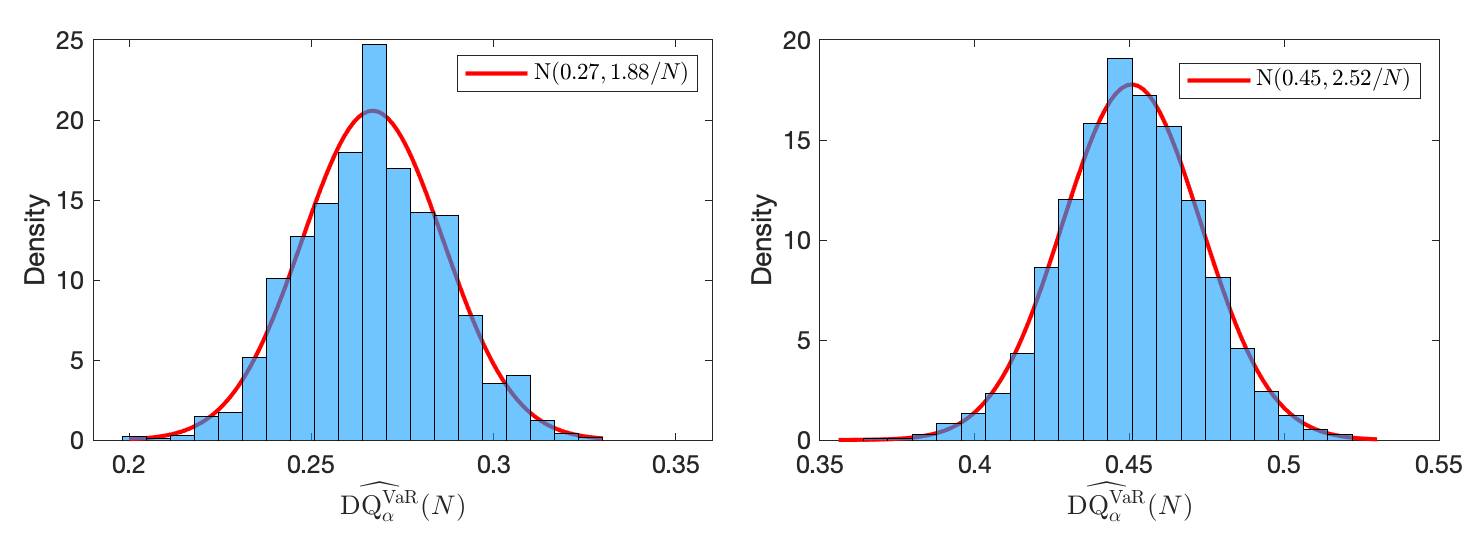}
 \captionsetup{font=small}
{\caption{   \small   Histogram of $\widehat{\mathrm{DQ}^{\VaR_\alpha}}(N)$  for $\mathbf X\sim  \mathrm N(\boldsymbol{\mu},\Sigma)$ (left panel) and $ \mathbf{Y}\sim \mathrm t(\nu,\boldsymbol{\mu},\Sigma)$ (right panel)}\label{fig:DQ_VaR}}
\end{figure}

\begin{figure}[htb!]
\centering
 \includegraphics[width=16cm]{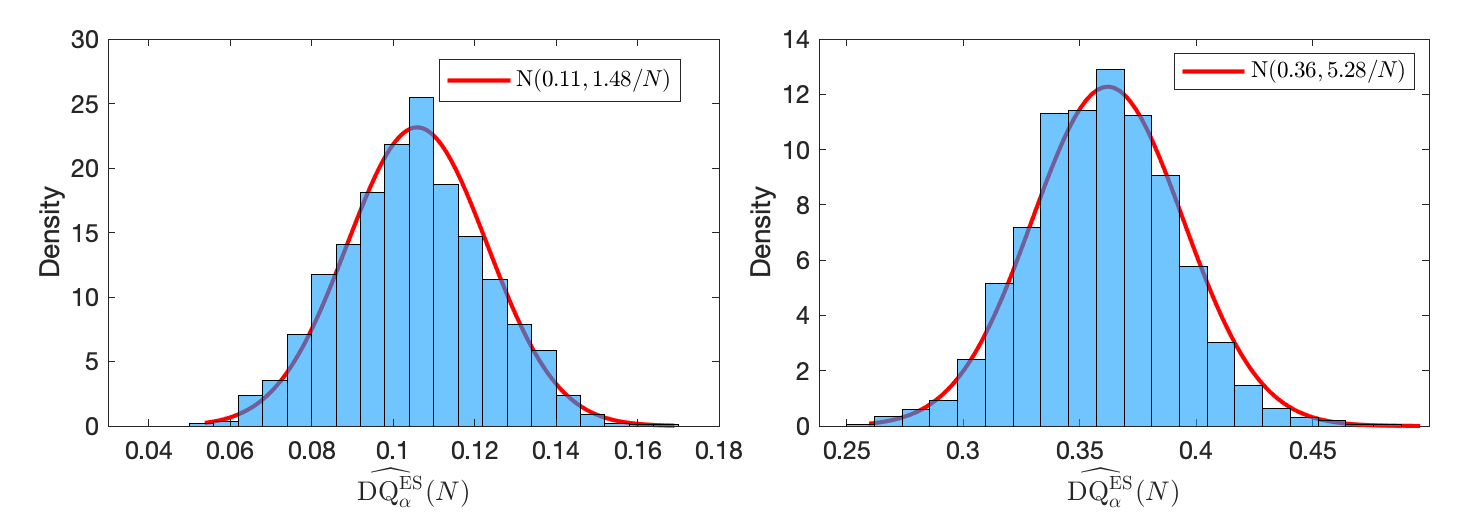}
 \captionsetup{font=small}
{\caption{   \small   Histogram of  $\widehat{\mathrm{DQ}^{\ES_\alpha}}(N)$ for $\mathbf X\sim  \mathrm N(\boldsymbol{\mu},\Sigma)$ (left panel) and $ \mathbf{Y}\sim \mathrm t(\nu,\boldsymbol{\mu},\Sigma)$ (right panel) }\label{fig:DQ_ES}}
\end{figure}

   As shown in  Figure \ref{fig:DQ_VaR} that,    the empirical estimates of  $\DQVaR$ matches quiet well with the density function of $\mathrm{N}(0.27,1.88/N)$ for  $\bf X\sim  \mathrm N(\boldsymbol{\mu},\Sigma)$,  and align closely with the density function of $\mathrm{N}(0.45,2.52/N)$ for $ \bf Y \sim \mathrm  t(\nu,\boldsymbol{\mu},\Sigma)$.
     Similarly, when $\bf X\sim  \mathrm N(\boldsymbol{\mu},\Sigma)$,   the empirical estimates of  $\DQES$ matches quiet well with the density function of  $\mathrm{N}(0.11,1.48/N)$, and  align closely with the density function of 
 $\mathrm{N}(0.36,5.28/N)$; see Figure \ref{fig:DQ_ES}.   However,  we do not expect the asymptotic variances of $\widehat{\mathrm{DQ}_\alpha^{\VaR}}(N)$ and $\widehat{\mathrm{DQ}_\alpha^{\ES}}(N)$ to have a consistent relationship in magnitude. As shown in the figures, the asymptotic variance of 
$\widehat{\mathrm{DQ}_\alpha^{\VaR}}(N)$ is smaller than that of $\widehat{\mathrm{DQ}_\alpha^{\ES}}(N)$  in the multivariate normal distribution case, but this relationship reverses when considering the multivariate t-distribution.
We find that the asymptotic variance of the multivariate t-distribution is larger than that of the multivariate normal distribution. This discrepancy can be attributed to the heavier tails of the t-distribution, which generally result in greater variability in the sample estimates.

 Next, we examine how the asymptotic variance $\sigma^2_\VaR$ and  $\sigma^2_\ES$, as derived in Theorems \ref{thm:2} and \ref{thm:3}, change with respect to   $\alpha$,  $r$ and  $n$ and $\nu$.  

\begin{figure}[htb!]
\centering
 \includegraphics[width=16cm]{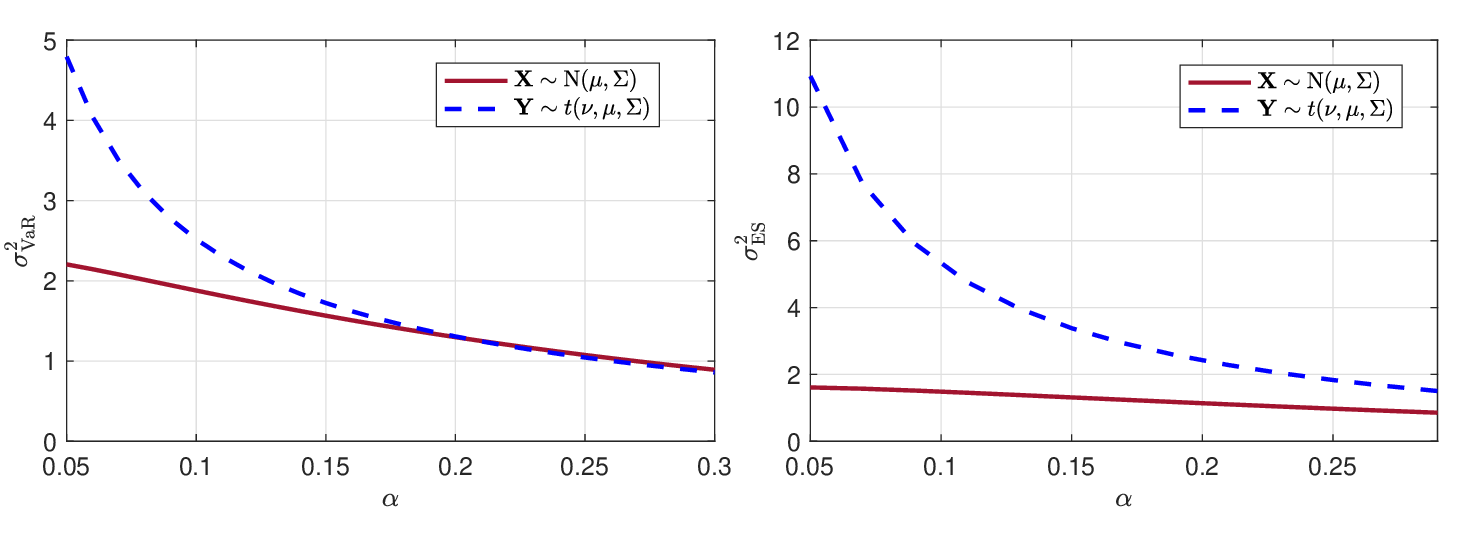}
 \captionsetup{font=small}
{\caption{   \small   Asymptotic variance    for $\alpha\in[0.05,0.3]$ with fixed $\nu=3$, $r=0.3$ and $n=5$}\label{fig:sigma_alpha}}
\end{figure}

In Figure \ref{fig:sigma_alpha}, we report  $\sigma^2_\VaR$ and  $\sigma^2_\ES$  for $\alpha\in[0.05,0.3]$ in the two models of $\mathbf X$ and $\mathbf Y$. The plot shows that both $\sigma^2_\VaR$ and  $\sigma^2_\ES$ decrease as $\alpha$ increases. This behavior aligns with the expectation that for smaller values of $\alpha$, fewer data points contribute to the estimation, resulting in higher variability.  
 Additionally, we note that the asymptotic variance of the multivariate t-distribution is larger than that of the multivariate normal distribution for both indices when $\alpha$ is relatively small, roughly for $\alpha<0.2$. 
\begin{figure}[htb!]
\centering
 \includegraphics[width=16cm]{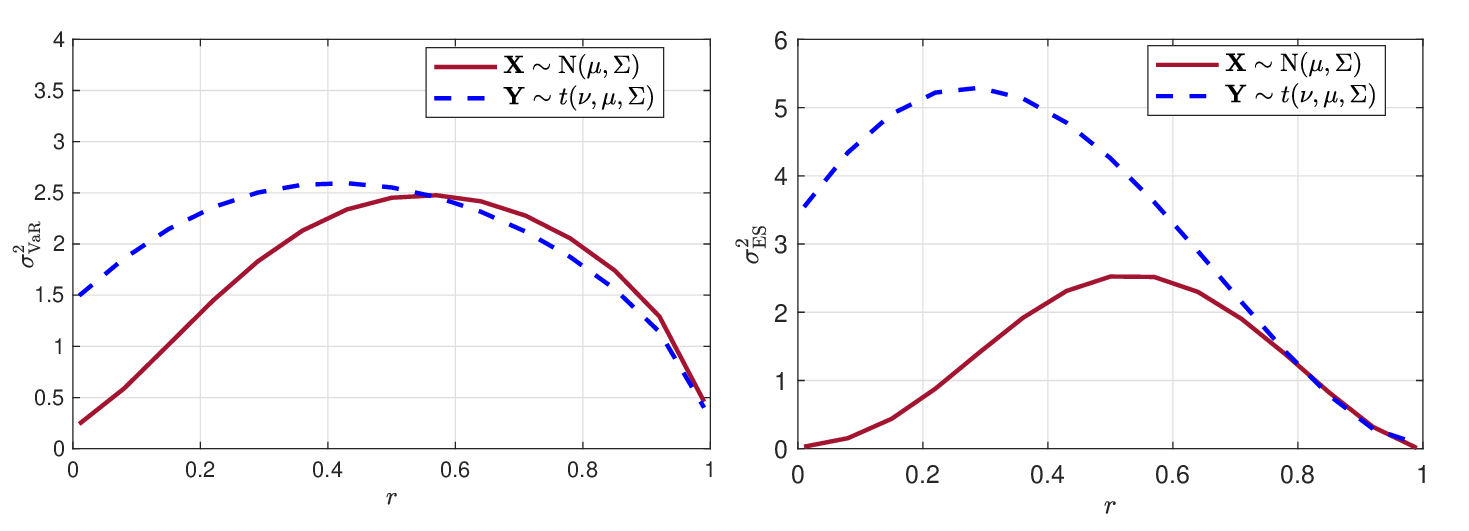}
 \captionsetup{font=small}
{\caption{   \small   Asymptotic variance   for $r\in[0.01,0.99]$ with fixed $\alpha=0.1$,  $\nu=3$ and  $n=5$}\label{fig:sigma_r}}
\end{figure}

In Figure \ref{fig:sigma_r}, we report how the asymptotic variance varies over  $r\in[0.1,0.99]$ in the two models. Both $\sigma^2_\VaR$ and $\sigma^2_\ES$ initially increase with $r$ and then decrease. This pattern reflects the combined effects of diversification and correlation among portfolio components.
When $r$ is small, the components are weakly correlated, resulting in smaller covariance contributions to $\Sigma_\VaR$ and $\Sigma_\ES$. As $r$ increases, stronger correlations amplify co-movements, raising the portfolio variance. When $r$ approaches 1, the portfolio behaves like a single asset, making its overall behavior more predictable from individual components.

\begin{figure}[htb!]
\centering
\includegraphics[width=16cm]{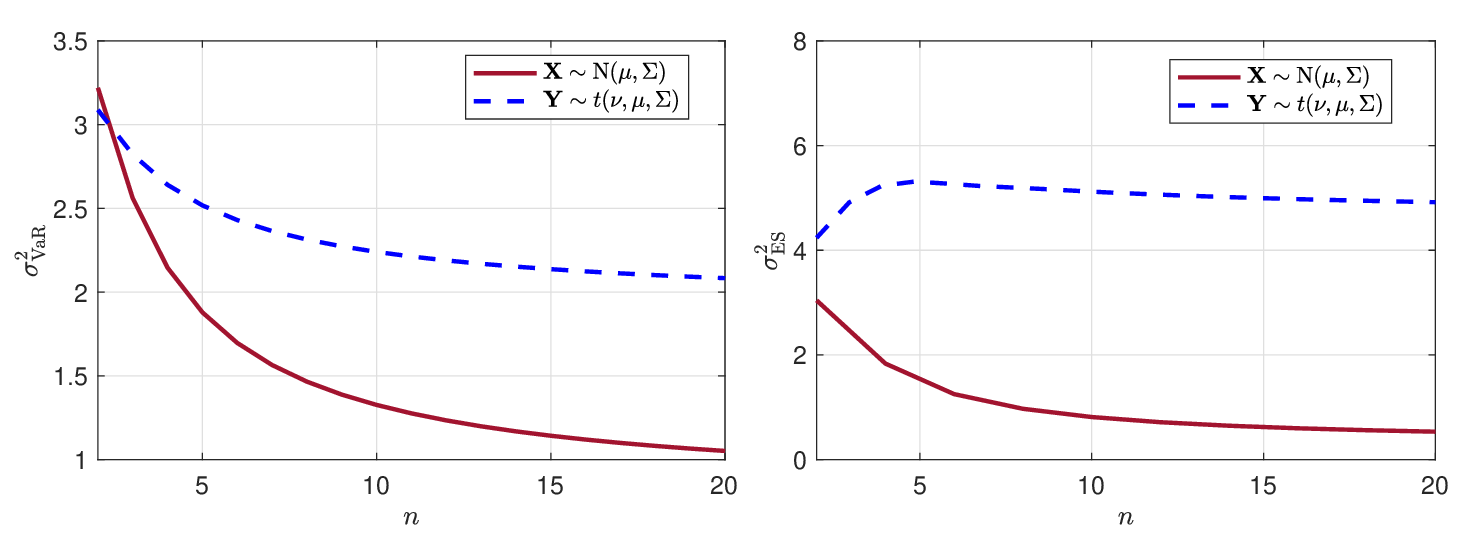}
 \captionsetup{font=small}
{\caption{   \small   Asymptotic variance    for $n\in[2,20]$ with fiexed $\alpha=0.1$, $\nu=0.3$ and $r=0.3$}\label{fig:sigma_stock}}
\end{figure}

In Figure \ref{fig:sigma_stock}, we report how the asymptotic variance changes with the number of assets $n \in [2, 20]$ under two models. For the Gaussian model $\mathbf{X} \sim \mathrm{N}(\boldsymbol{\mu}, \Sigma)$, both $\sigma^2_{\VaR}$ and $\sigma^2_{\ES}$ steadily decrease as $n$ increases. For the multivariate t-distribution $Y \sim \mathrm{t}(\nu, \boldsymbol{\mu}, \Sigma)$, the asymptotic variance $\sigma^2_{\ES}$ is not necessarily monotonic in $n$. However, once $n$ becomes moderately large ($n > 5$), a clear decreasing trend emerges. 


\begin{figure}[htb!]
\centering
\includegraphics[width=16cm]{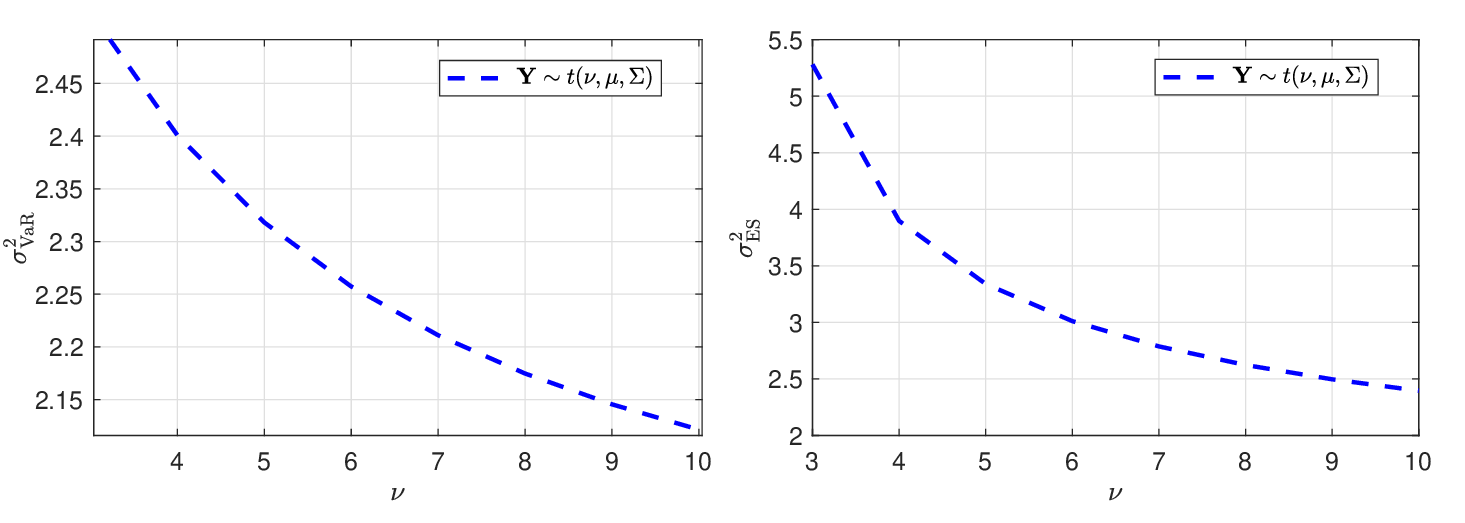}
 \captionsetup{font=small}
{\caption{   \small   Asymptotic variance  for $\nu\in[3,10]$ with fixed $\alpha=0.1$, $r=0.3$ and $n=5$}\label{fig:sigma_nu}}
\end{figure}

 Figure \ref{fig:sigma_nu} present the values of $\sigma^2_\VaR$ and  $\sigma^2_\ES$ for the t-models with varying $\nu$,  where $\nu\in[3,10]$. We observe a monotonic relation that both $\sigma^2_\VaR$ and  $\sigma^2_\ES$ are decreasing in $\nu$.  This is because the distribution's tails become lighter as $\nu$ increases, reducing the influence of extreme values and stabilizing the sample mean.  

In the following context, building on the results from Proposition \ref{prop:5}, we examine the asymptotic normality of $\widehat
{\mathrm{DR}^{\VaR_\alpha}}(N)$ and $\widehat{\mathrm{DR}^{\ES_\alpha}}(N)$
for the same models of $\mathbf{X} \sim\mathrm N(\boldsymbol{\mu},\Sigma)$ and $ \mathbf{Y}\sim \mathrm  t(\nu,\boldsymbol{\mu},\Sigma)$, using the same parameter settings as in Figures \ref{fig:DQ_VaR} and \ref{fig:DQ_ES}. Recall from Section 5.1 of \cite{HLW23} that for a centered elliptical distribution, both  ${\rm DR}^{\VaR_\alpha}$ and ${\rm DR}^{\ES_\alpha}$  do not depend on $\tau$, $\alpha$ or whether the risk measure is VaR or ES,  and are always equal to $1/k_{\Sigma}$.
	More precisely, for $\mathbf X \sim  \mathrm{E}_{n}(\mathbf{0}, \Sigma, \tau)$ and $\alpha\in(0,1/2)$, we have
	\begin{equation*}
		\label{eq:DRellip}
		{\rm DR}^{\VaR_\alpha}(\mathbf X)={\rm DR}^{\ES_\alpha}(\mathbf X)= {1}/{k_\Sigma}.
	\end{equation*}

\begin{figure}[htb!]
\centering
 \includegraphics[width=16cm]{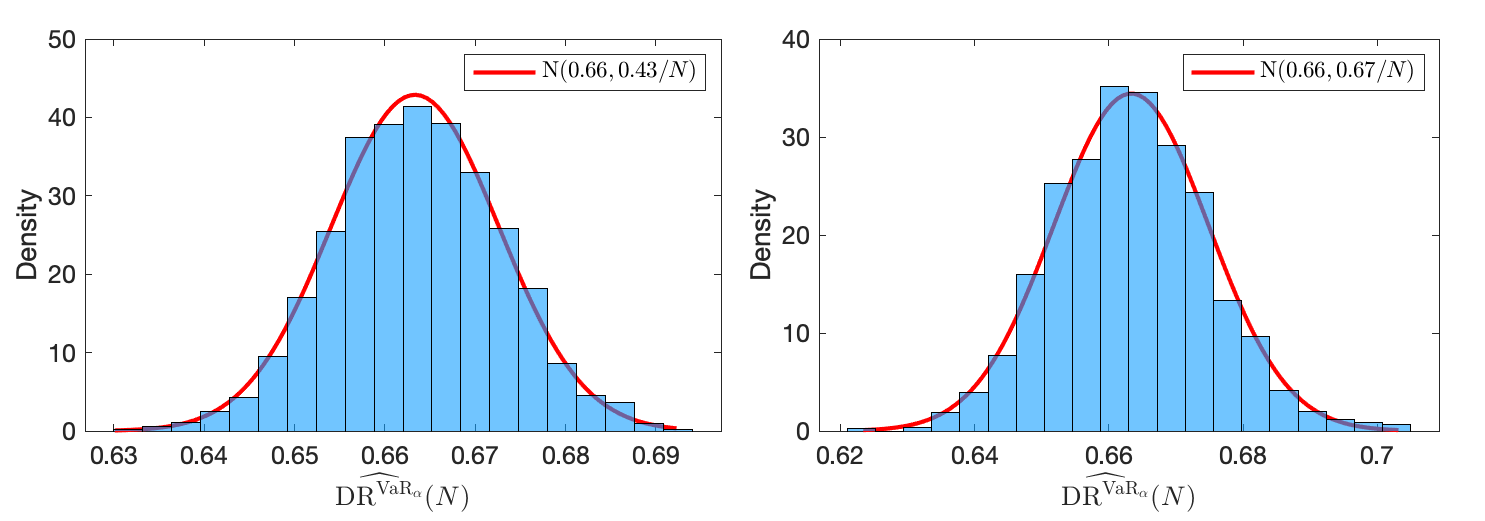}
 \captionsetup{font=small}
{\caption{   \small   Histogram of $\widehat{\mathrm{DR}^{\VaR_\alpha}}(N)$  for $\mathbf X\sim  \mathrm N(\boldsymbol{\mu},\Sigma)$ (left panel) and $ \mathbf{Y}\sim \mathrm t(\nu,\boldsymbol{\mu},\Sigma)$ (right panel)}\label{fig:DR_VaR}}
\end{figure}

\begin{figure}[htb!]
\centering
 \includegraphics[width=16cm]{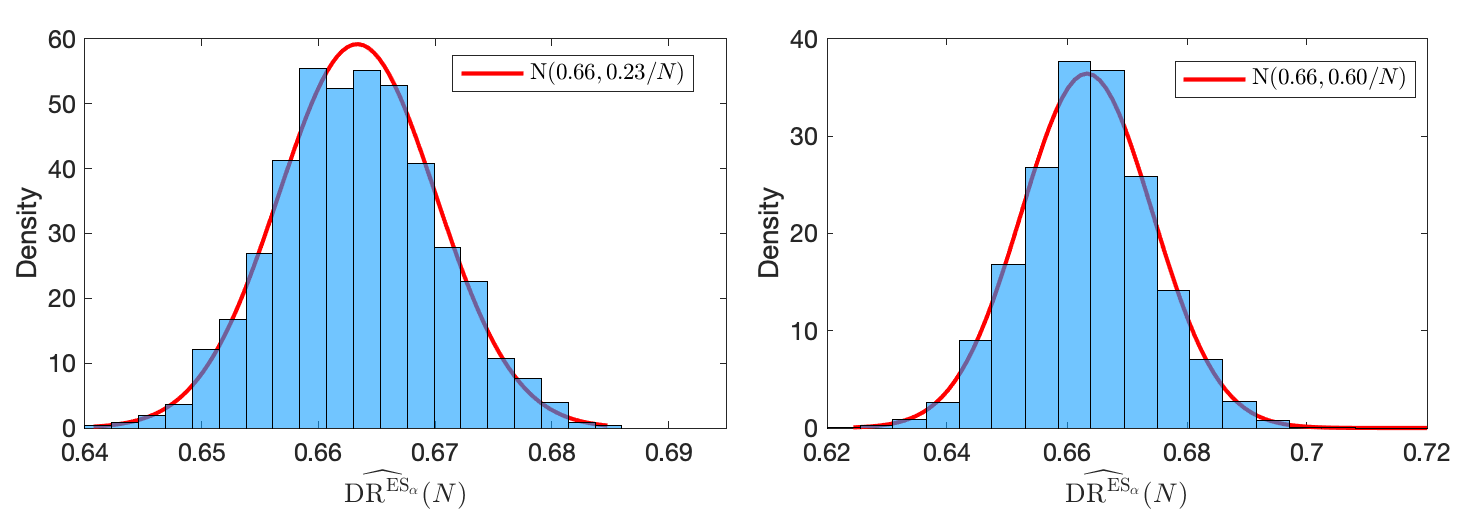}
 \captionsetup{font=small}
{\caption{   \small   Histogram of  $\widehat{\mathrm{DR}_\alpha^{\ES}}(N)$ for $\mathbf X\sim  \mathrm N(\boldsymbol{\mu},\Sigma)$ (left panel) and $ \mathbf{Y}\sim \mathrm t(\nu,\boldsymbol{\mu},\Sigma)$ (right panel) }\label{fig:DR_ES}}
\end{figure}
  Figure \ref{fig:DR_VaR} shows that the empirical estimates of  $\mathrm{DR}^{\VaR_\alpha}$ align quite well with the density function of  $\mathrm{N}(0.66,0.43/N)$ for  $\bf X\sim  \mathrm N(\boldsymbol{\mu},\Sigma)$,  and  closely match the density function of $\mathrm{N}(0.66,0.67/N)$ for $ \bf Y \sim \mathrm  t(\nu,\boldsymbol{\mu},\Sigma)$.
     Similarly,   Figure \ref{fig:DR_ES} shows that  the empirical estimates of  $\mathrm{DR}^{\ES_\alpha}$ matches quiet well with the density function of  $\mathrm{N}(0.66,0.23/N)$ for   $\bf X\sim  \mathrm N(\boldsymbol{\mu},\Sigma)$,  and  align closely with the density function of 
 $\mathrm{N}(0.66,0.60/N)$ for $ \bf Y \sim \mathrm  t(\nu,\boldsymbol{\mu},\Sigma)$. 
We also observe that the asymptotic variance of DR under both models is smaller than that under DQ, which appears to be a statistically favorable property. However, it is important to note that DR is not location-invariant. As discussed in Section \ref{sec:6}, the empirical estimator for DR based on a cash-additive risk measure fails when considering a centralized portfolio loss  $\bX-(\phi(X_1), \dots, \phi(X_n))$  as the variance can become explosive in such cases.
 
 Next, we present $\widetilde{\sigma}^2_\VaR$ and $\widetilde{\sigma}^2_\ES$  for the model  $\mathbf{X} - (\phi(X_1), \dots, \phi(X_n)) + \epsilon$  with different $\epsilon$ where $\mathbf X\sim \mathrm N(\boldsymbol{\mu},\Sigma)$ in Figure \ref{fig:DR_VaR_n_eps}. We observe that the asymptotic variance of both $\widehat
{\mathrm{DR}^{\VaR_\alpha}}(N)$ and $\widehat{\mathrm{DR}^{\ES_\alpha}}(N)$ can be very large when $\epsilon$ is approach to 0. However, this issue does not arise for the DQ estimator, as it is location-invariant.

\section{Real data analysis}\label{sec:real_data}
\begin{figure}[htb!]
\centering
 \includegraphics[width=16cm]{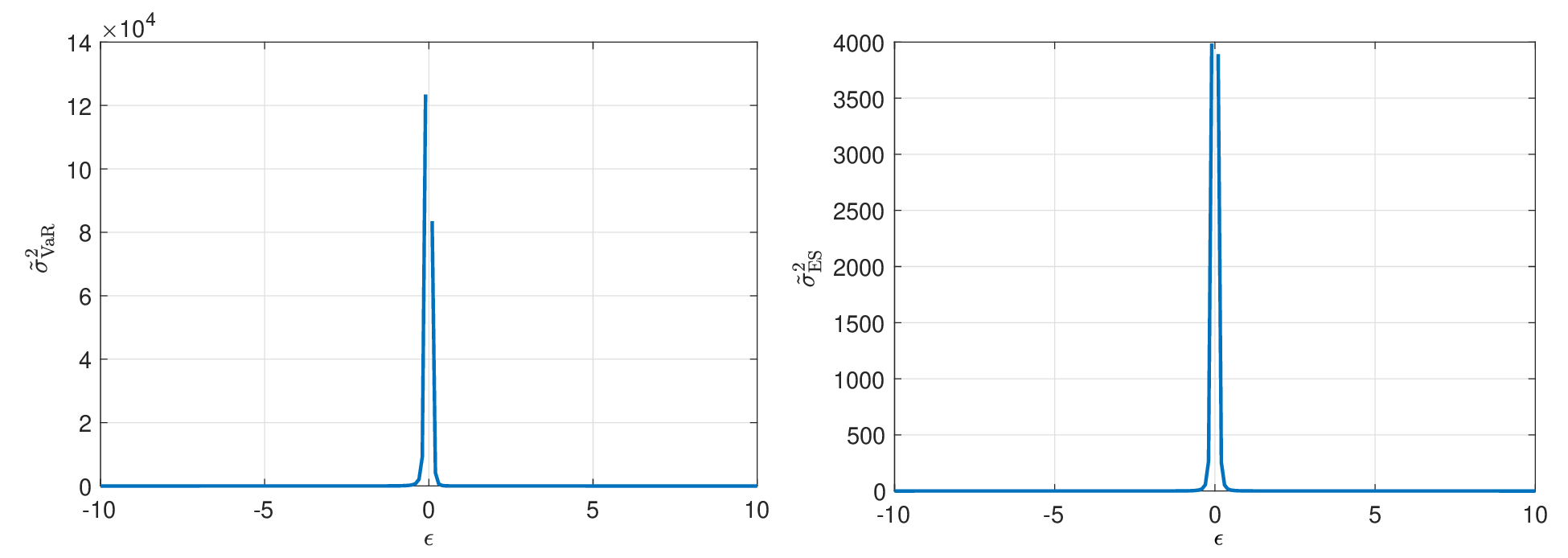}
 \captionsetup{font=small}
{\caption{   \small   The values of $\widetilde \sigma^2_\VaR$ and $\widetilde \sigma^2_\ES$  for $\epsilon\in[-10,10]$ with $\alpha=0.1$, $r=0.3$ and $n=5$}\label{fig:DR_VaR_n_eps}}
\end{figure}

The proposed nonparametric estimator is now investigated for financial data. Let $X_{t,j}$ denote the log-loss of asset $j$ on day $t$, where $t=1,\dots,N$ and $j=1,\dots,n$.  
For each day $t \ge 1$ and each asset $j$, we estimate $\DQVaR$ and $\DQES$ using the empirical estimator in Section~\ref{sec:4}, based on a rolling window of the most recent $m$ observations $\{X_{s,j}\}_{s=t-m}^{t-1}$. 
For the days $t < m$, the rolling window can be adjusted to include available past observations.

To construct confidence intervals for $\DQVaR$ and $\DQES$, one could, in principle employ a blockwise bootstrap for mixing sequences. However, since quantiles are involved, such nonparametric inference is often inefficient due to the limited number of available blocks. A more effective approach is to model the temporal dependence of each series $\{X_{s,j}\}_{s=t-m}^{t-1}$ using a suitable time series model and then apply a residual-based bootstrap, as suggested by \cite{APWY19} and \cite{LW23}.

Specifically, for each asset $j$, we fit an AR(1)--GARCH(1,1) model with Student-$t$ innovations on the rolling window via quasi-maximum likelihood estimation (see, e.g., \cite{FZ04}):
\begin{align}\label{AG}
X_{s,j} = c_j + \phi_j X_{s-1,j} + \varepsilon_{s,j}, ~~
\varepsilon_{s,j} = \sigma_{s,j} z_{s,j}, ~~
\sigma_{s,j}^2 = \omega_j + \alpha_j \varepsilon_{s-1,j}^2 + \beta_j \sigma_{s-1,j}^2,
\end{align}
where $z_{s,j}\sim t_{\nu_j}(0,1)$.  
The parameter estimates are
$
\widehat{\boldsymbol{\theta}}_j = (\widehat{c}_j, \widehat{\phi}_j, \widehat{\omega}_j, \widehat{\alpha}_j, \widehat{\beta}_j, \widehat{\nu}_j).
$
Standardized residuals are
$$
\widehat{z}_{s,j} = \frac{X_{s,j} - \widehat{c}_j - \widehat{\phi}_j X_{s-1,j}}{\widehat{\sigma}_{s,j}}, 
\quad s = t-m, \dots, t-1,
$$
with $\widehat{\sigma}_{s,j}^2$ initialized at the unconditional variance if $\widehat{\alpha}_j+\widehat{\beta}_j < 1$, and at the sample variance otherwise.  
Collecting standardized residuals across assets yields
$
\widehat{\boldsymbol{z}}_s = (\widehat{z}_{s,1}, \dots, \widehat{z}_{s,n}).
$ 

To generate bootstrap samples while preserving cross-sectional dependence, we resample with replacement from the standardized residuals in the rolling window, $\{\widehat{\boldsymbol{z}}_s\}$, using the same resampled index for all assets. 
Specifically, for each bootstrap replication, we draw an index vector $\mathcal{I} = (i_1, \dots, i_m)$ with replacement from the rolling window, and then set 
$\widehat{\boldsymbol{z}}^*_{k} = \widehat{\boldsymbol{z}}_{i_k}$ for $k=1,\dots,m$. For each asset $j$, we reconstruct the bootstrap paths ${X^\star_{s,j}}$ by iterating model \eqref{AG} with the estimated parameters $\widehat{\boldsymbol{\theta}}_j$. Using these bootstrap samples, we compute the bootstrapped estimators $\widehat{\DQVaR}^*(m)$ and $\widehat{\DQES}^*(m)$.\footnote{The superscript $*$ denotes quantities generated in the bootstrap procedure rather than observed from the original data.} Repeating this procedure yields a collection of simulated $\widehat{\DQVaR}^*(m)$ and $\widehat{\DQES}^*(m)$ values. The 95\% percentile-based bootstrap confidence interval (CI) is then obtained from the 2.5\% and 97.5\% quantiles of these simulated estimators.

We construct the portfolio by selecting the largest stocks by market capitalization in 2010 from different S\&P 500 sectors: XOM (Energy), MSFT (Information Technology), BRK/B (Financials), WMT (Consumer Staples), and JNJ (Health Care). Using daily log-loss data from January 4, 2010, to August 29, 2025 (3,939 observations), we compute  DQ on a monthly basis, using a rolling window of $m=500$ trading days  (approximately 21 trading days per month), and assess statistical reliability via a residual-based bootstrap with 500 replications.

\begin{figure}[htb!]
\centering
 \includegraphics[width=14cm]{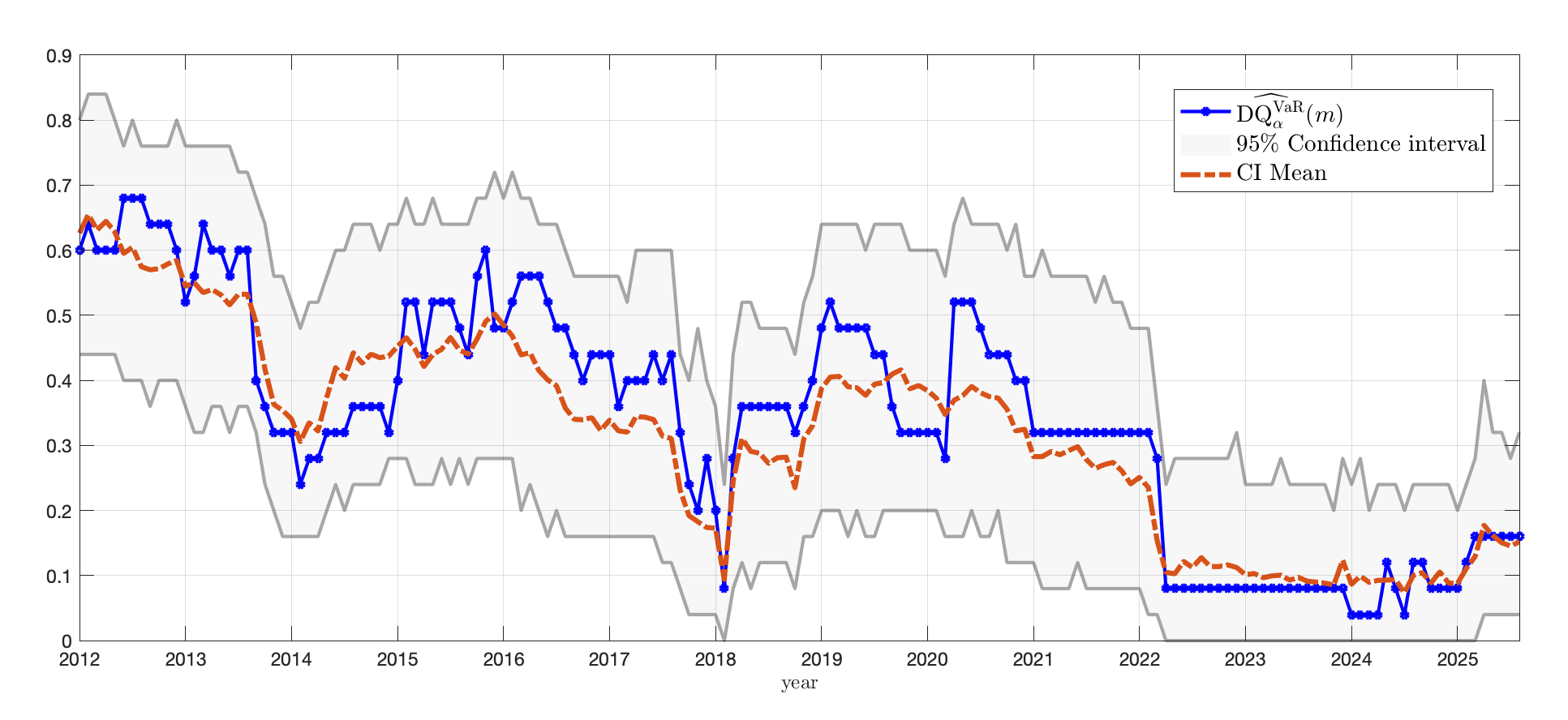}
 \captionsetup{font=small}
{\caption{   \small Bootstrap confidence bands for $\DQVaR$ with $\alpha=0.05$}\label{fig:DQ_VaR_0.05}}
\end{figure}

\begin{figure}[htb!]
\centering
 \includegraphics[width=14cm]{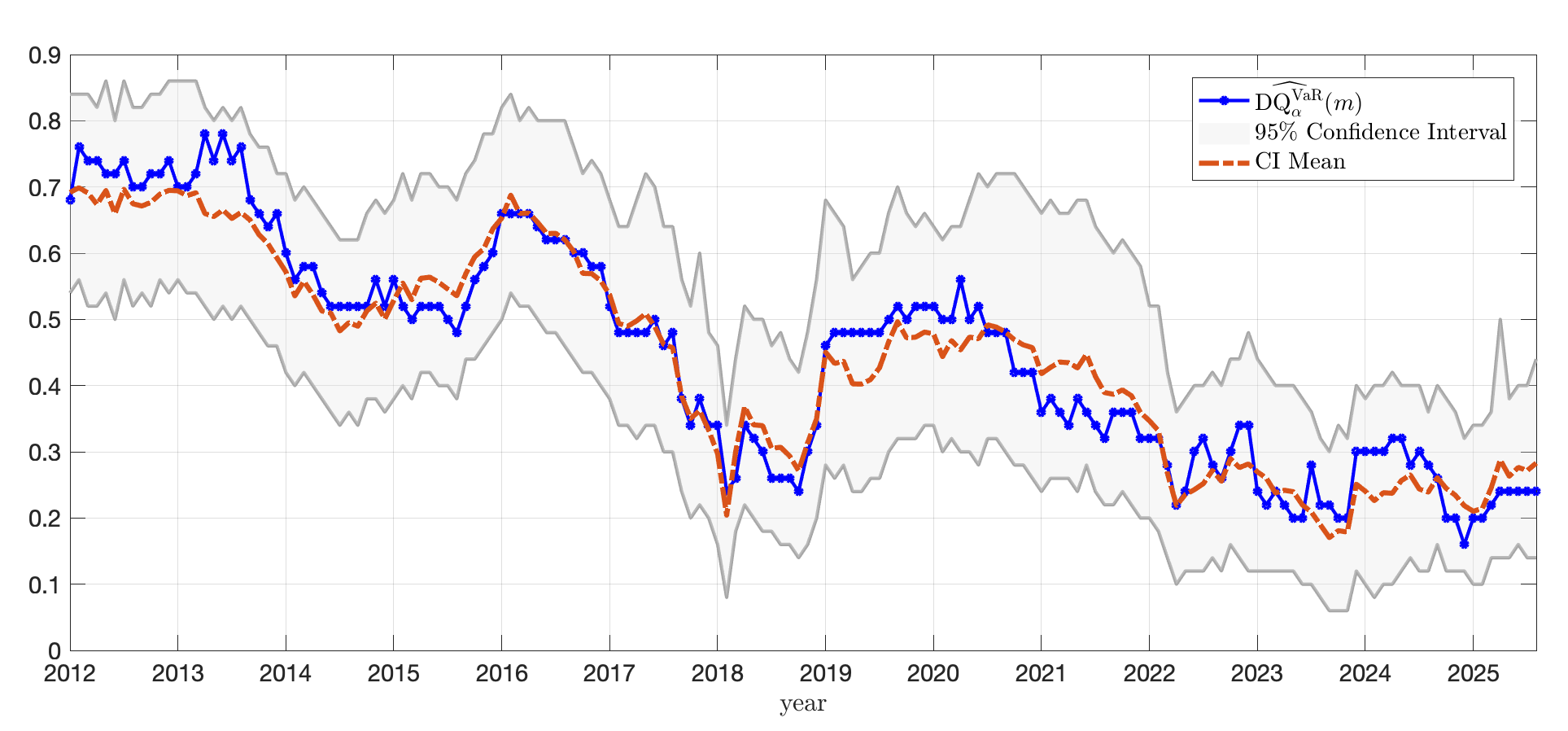}
 \captionsetup{font=small}
{\caption{   \small    \small Bootstrap confidence bands for $\DQVaR$ with $\alpha=0.1$}\label{fig:DQ_VaR_0.1}}
\end{figure}
Figures~\ref{fig:DQ_VaR_0.05} and \ref{fig:DQ_VaR_0.1} show the 95\% CIs, bootstrap means, and empirical estimates of $\DQVaR$ based on 500 trading days at $\alpha=0.05$ and $\alpha=0.1$. The bootstrap means lie near the center of the CIs, suggesting a roughly symmetric distribution of the bootstrap estimates. Moreover, the empirical estimates align closely with the bootstrap means for both $\mathrm{DQ}^{\VaR}_{0.05}$ and $\mathrm{DQ}^{\VaR}_{0.1}$, indicating that the estimator does not systematically over- or understate the underlying risk measure. These results underscore the estimator’s combined strengths of statistical reliability and economic interpretability, supporting its use in portfolio allocation and risk assessment.

We also compute the variance by averaging the bootstrap variances across replications. At $\alpha = 0.05$, the variance is 0.0096, while at $\alpha = 0.1$ it decreases to 0.0066. This pattern is consistent with Figure \ref{fig:sigma_alpha}, as fewer observations contribute to the estimation of $\DQVaR$ when $\alpha$ is small. Furthermore, during the COVID-19 period, the confidence intervals are noticeably wider, reflecting heightened volatility and more extreme events that increase the variability of the estimates.

\begin{figure}[htb!]
\centering
 \includegraphics[width=14cm]{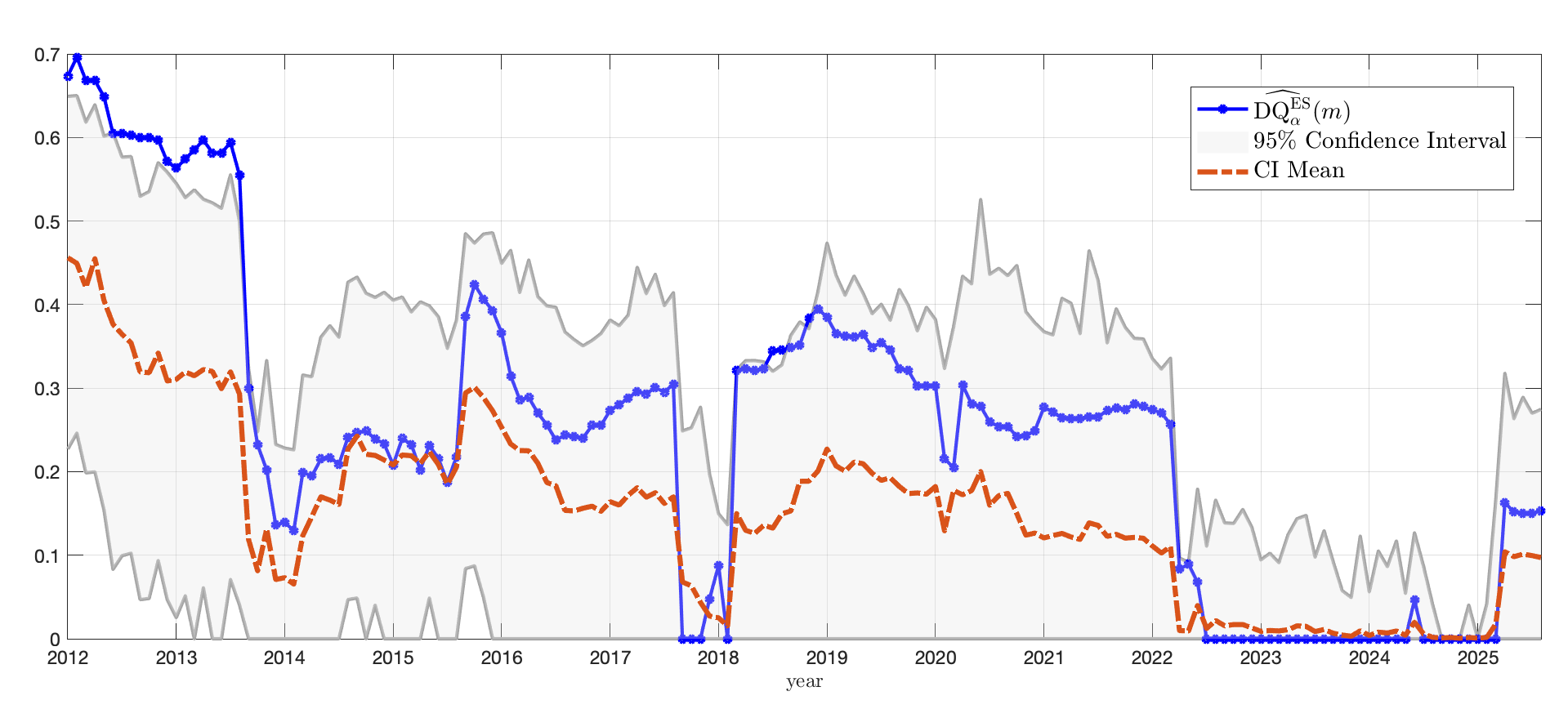}
 \captionsetup{font=small}
{\caption{   \small    Bootstrap confidence bands for $\DQES$ with $\alpha=0.05$}\label{fig:DQ_ES_0.05}}
\end{figure}

\begin{figure}[htb!]
\centering
 \includegraphics[width=14cm]{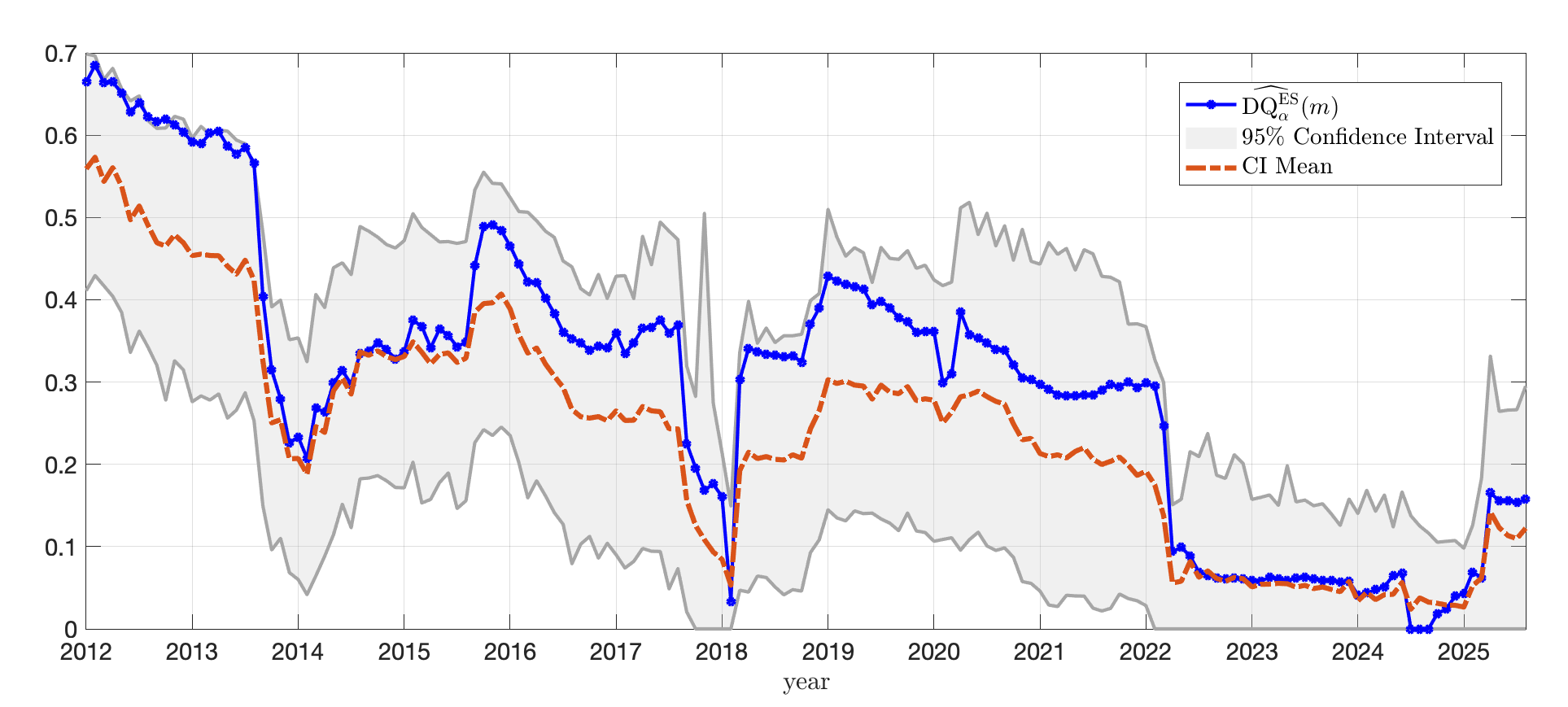}
 \captionsetup{font=small}
{\caption{   \small   Bootstrap confidence bands for $\DQES$ with $\alpha=0.1$}\label{fig:DQ_ES_0.1}}
\end{figure}
Figures~\ref{fig:DQ_ES_0.05} and \ref{fig:DQ_ES_0.1} show the 95\% CIs, bootstrap means, and empirical estimates of $\DQES$ based on 500 trading days at $\alpha=0.05$ and $\alpha=0.1$. Similar to the $\DQVaR$ case, the bootstrap mean lies near the center of the CIs, indicating a roughly symmetric distribution of the bootstrap estimates. The CIs are relatively wider at $\alpha=0.05$ due to limited tail data, with the lower bound occasionally reaching zero. Between 2012–2014, some empirical estimates fall outside their bootstrap CIs, and several lie close to the upper bound, especially during periods of weak diversification as reflected by $\DQES$. This arises because the empirical ES incorporates all extreme past losses, which tends to overestimate tail risk compared with the simulated estimator; the effect is more pronounced on aggregated losses when the market is less diversified.


We further compute the variance of $\DQES$ across bootstrap replications. For $\alpha=0.05$, the variance is 0.0089, while for $\alpha=0.1$, it decreases to 0.0060. This confirms our observation from Figure \ref{fig:sigma_alpha} that the variability of $\DQES$ declines as $\alpha$ increases.

\section{Conclusion}\label{sec:8}
This paper investigates the asymptotic behavior of empirical DQ estimators when the underlying risk measure is law-invariant, with a particular emphasis on VaR and ES. We establish both consistency and asymptotic normality for empirical $\DQVaR$ and $\DQES$ estimators under i.i.d.~and $\alpha$-mixing settings, which facilitate the application of DQ in realistic financial environments where dependence structures and volatility clustering are present. Compared with the estimation of DR, our DQ is more robust in location. We provide an example with real-world financial data to illustrate the confidence intervals of DQ estimators through AR-GARCH modeling and residual-based bootstrap inference—demonstrates.



Looking ahead, several promising directions remain for future research. One avenue is to integrate DQ as a constraint in portfolio optimization. Such an approach would not only balance risk minimization and return maximization but also explicitly ensure a desirable level of diversification, thereby offering a more comprehensive decision-making framework. Another avenue lies in extending the axiomatic theory of DQ. The current axiomatic framework is built on the premise that the index is defined with respect to a given class of parametric risk measures. A natural and challenging question is whether DQ can be characterized without relying on any pre-specified risk measure. Developing such a more general axiomatization would not only deepen our understanding of DQ, but also broaden its applicability in risk measurement and portfolio analysis.

    \subsection*{Acknowledgments}
 We thank  Junyi Guo and Ruodu Wang 
for helpful comments and discussions on an earlier version of the paper. Xia Han is supported by the National Natural Science Foundation of China (Grant Nos. 12301604, 12371471, and 12471449). Liyuan Lin is supported by the Department Seed Grant of Monash University (Grant No.~1752540) and 2025 ECA Research Grant of Monash University (Grant No.~1750890).

\appendix
\renewcommand{\thelemma}{A.\arabic{lemma}}
\renewcommand{\theequation}{A.\arabic{equation}}

\section{A further  discussion of DQ based on expectiles}\label{App:A}
In addition to VaR and ES, expectiles have emerged as an important class of risk measures. Given their significance, \cite{HLWW24} extends the DQ framework by incorporating expectiles, thereby broadening the toolkit available for risk management beyond VaR and ES.   In this section, we will further examine  the asymptotic normality of DQ based on expectiles.  

The expectile at a given confident level $\alpha\in(0,1)$ for a loss random variable is the unique minimizer of an asymmetric quadratic loss. Mathematically, the expectile   of a loss random variable $X\in L^2$  at a confidence level $\alpha \in(0,1)$, denoted by $\ex_{\alpha}(X)$, is defined as the following minimizer:  \begin{equation*}\label{eq:expectile}
\ex_{\alpha}(X)=\argmin_{x \in \mathbb{R}}   \left\{(1-\alpha) \mathbb{E}\left[(X-x)_{+}^{2}\right]+\alpha \mathbb{E}\left[(X-x)_{-}^{2}\right]\right\}
\end{equation*}
where $(x)_+=\max\{0,x\}$ and  $(x)_+=\max\{0,-x\}$.
 By the first-order condition, $\ex_{\alpha}(X)$  is equal to the unique number $t$   satisfying 
\begin{equation}\label{eq:ex0} (1-\alpha) \mathbb{E}\left[(X-t)_{+}\right]=\alpha \mathbb{E}\left[(X-t)_{-}\right].
\end{equation}

The equation \eqref{eq:ex0} is well-defined for each $X \in L^1$, which serves as the natural domain for expectiles. We take \eqref{eq:ex0} as the definition of expectiles, and let $\X=L^1$ in the following context. 

Let $\bX^{(1)}, \dots, \bX^{(N)}$ be i.i.d.~samples of $\bX=(X_1, \dots, X_n)$ where $\bX^{(k)}=(X_1^{(k)}, \dots, X_n^{(k)})$ for $k \in [N]$. For $i\in [n]$, let  $\widehat x^{\ex_\alpha}_i$ be the empirical estimators for  $\ex_\alpha(X_i)$.
 For a given $\alpha \in(0,1/2)$,  the empirical estimator for $\DQex$ is given by 
\begin{equation*}
\widehat \DQex(N) =\frac{1}{\alpha}
\frac{\sum_{k=1}^N\left[ (S^{(k)}-t) _+\right]}{\sum_{k=1}^N\left[\left| (S^{(k)}-t) \right|\right]},\mbox{~where }S^{(k)}=\sum_{i=1}^n X^{(k)}_i \mbox{~and~}t=\sum_{i=1}^n \widehat x^{\ex_\alpha}_i.
\end{equation*}
For empirical estimators of expectile-based risk measures, we refer, for example, to \cite{KZ17} and \cite{MBL21}.

For a random variable $X\sim F$, we define  \begin{equation}\label{eq:tildeF}
 \widetilde{F}(y) = \frac{y F(y) - \int_{-\infty}^{y} x \, \d F(x)}{2 \left( y F(y) - \int_{-\infty}^{y} x \, \d F(x) \right) + \mathbb{E}[X] - y}. \end{equation}
  It is easy to check that $\widetilde F$ is a continuous distribution function that strictly increases its support. In this section, the tilde notation is used consistently in the following context, with each function following a similar transformation  as defined by \eqref{eq:tildeF}.
 
 Recall that $S\sim G$. Theorem 1 of \cite{HLWW24} shows that  ${\rm DQ}^{\ex}_\alpha (\mathbf X)$ can also be computed by
\begin{equation}\label{eq:DQ_ex}
\mathrm{DQ}^{\ex}_\alpha (\mathbf{X})=\frac{1-\widetilde{G}\left(\sum_{i=1}^n \ex_\alpha(X_i)\right) }{\alpha},~~~\mathbf X\in \X^n.\end{equation}
 Thus, \eqref{eq:DQ_ex} can be rewritten as 
$$
\DQex (\mathbf X) =\frac{1}{\alpha} \mathbb Q\left(\sum_{i=1}^{n}X_i> \sum_{i=1}^n \ex_{\alpha}(X_i)\right),~~~\mathbf X\in \X^n, 
$$
for some probability measure $\mathbb{Q}$,
which reveals  that $\mathrm{DQ}^{\ex}_\alpha(\mathbf{X})$ shares a structural similarity with $\mathrm{DQ}^{\VaR}_\alpha(\mathbf{X})$ given by \eqref{eq:DQ_VaR}.   While $\widetilde{G}$ is indeed a distribution function with a positive density function, the estimator $\widetilde{\widehat{G}^N}$, constructed from  i.i.d.~samples of $X_i$ for $i=1,\dots,n$, does not serve as  an empirical distribution function for  $\widetilde{G}$. This discrepancy arises because $\widetilde{G}$ represents the distribution of a distinct random variable; hence, an empirical distribution function for $\widetilde{G}$ would need to be derived from samples of that specific variable.  Nonetheless, the following proposition establishes a result analogous to Bahadur’s representation.

\begin{lemma}
\label{lemma:ex}
Let $X^{(1)},\dots,X^{(N)}$ be an i.i.d.~sample from $X\in L^1$ with $F$ as distribution function of $X$, and ${F}^N$ be the empirical distribution based on $X^{(1)},\dots,X^{(N)}$. 
Let $\widetilde{  F^N}$\footnote{We simplify the notation by using $\widetilde{  F^N}$  instead of $\widetilde{ \widehat F^N}$.}  and $\widetilde{F}$ be defined as in \eqref{eq:tildeF}, with  the positive density functions $\widetilde{f^N}$ and $\widetilde f$.
  Then as $k\to\infty$, we have
$$
\widetilde{  F^N}^{-1}(p)=\widetilde{F}^{-1}(p)+\frac{p-\widetilde{  F^N}(\widetilde{F}^{-1}(p))}{\widetilde f(\widetilde{F}^{-1}(p))}+\smallO_\p\left(\frac{1}{\sqrt{N}}\right),~~~~p\in(0,1).
$$
\end{lemma}
\begin{proof}
For $p\in(0,1)$ and $t\in \mathbb{R}$, let $\xi^N(t)=\widetilde{F}^{-1}(p)+t/\sqrt{N}$ and
$$
Z^N(t)=\frac{\sqrt{N}\left( \widetilde{F}(\xi^N(t))-\widetilde {F^N}(\xi^N(t))\right) }{\widetilde f(\widetilde{F}^{-1}(p))}, \; U^N(t)=\frac{\sqrt{N}\left( \widetilde{F}(\xi^N(t))-p\right) }{\widetilde f(\widetilde{F}^{-1}(p))}.
$$
For a fixed $t \in \R$, take $N$ large enough such that $\widetilde{F}^{-1}(p) \in (X_{[k]}, X_{[k+1]})$ for some $k \in [N]$ where $X_{[k]}$ is the $k$-th order statistics of $\{X^{(k)}\}_{k \in [N]}$ and $t/\sqrt{N}$ small enough, 
 by the  mean value theorem,  we have 
\begin{equation*}
\begin{aligned}
Z^N(t)-Z^N(0)&=\frac{\sqrt{N}\left( \widetilde{F}(\xi^N(t))-\widetilde {F^N}(\xi^N(t))-p+\widetilde {F^N}(\widetilde{F}^{-1}(p))\right) }{\widetilde f(\widetilde{F}^{-1}(p))}\\
&=\frac{\sqrt{N}\left( \left( \widetilde{f}(\widetilde{F}^{-1}(p))+\smallO_\p(1)\right) t/\sqrt{N}-\left( \widetilde{f^N}(\widetilde{F}^{-1}(p))+\smallO_\p(1)\right) t/\sqrt{N}\right) }{\widetilde f(\widetilde{F}^{-1}(p))}\\
&=t\left( 1-\frac{\widetilde {f^N}(\widetilde{F}^{-1}(p))}{\widetilde  f(\widetilde{F}^{-1}(p))}+\smallO_\p(1)\right). 
\end{aligned}
\end{equation*}
By direct calculation, we have
\begin{equation*}\label{eq:wideF} \widetilde{f}(y)=\frac{F(y)\mathbb{E}[X]-\int_{-\infty}^{y}x\d F(x)}{\left( 2\left( yF(y)-\int_{-\infty}^{y}x\d F(x)\right)+\mathbb{E}[X]-y\right) ^2}, 
\end{equation*}
and
 \begin{equation*} \widetilde{f^N}(y)=\frac{\wF^N(y)\displaystyle{\frac{\sum_{k=1}^N X^{(k)}}{N}}-\displaystyle{\frac{1}{N}\sum_{k=1}^N X^{(k)}\id_{\{X^{(k)} \le y\}}}}{\left( 2\left( y\wF^N(y)-\displaystyle{\frac{1}{N}\sum_{k=1}^N X^{(k)}\id_{\{X^{(k)} \le y\}}}\right)+\displaystyle{\frac{\sum_{k=1}^N X^{(k)}}{N}}-y\right) ^2}. \end{equation*}

As $N \to \infty$, we have $\wF^N(y) \to F(y)$, $\sum_{k=1}^N X^{(k)} \to \E[X]$ and $\sum_{k=1}^N X^{(k)}\id_{\{X^{(k)} \le y\}}/N \to \int_{-\infty}^{y}x\d F(x)$, $\p$- a.s. Thus,
$$\frac{\widetilde{f^N}(\widetilde{F}^{-1}(p))}{\widetilde{f}(\widetilde{F}^{-1}(p))}\to 1, ~~\mathbb{P}\text{-a.s.}$$ 
Therefore, for a fixed $t\in\R$, $Z^N(t)-Z^N(0)=\smallO_\p(1)$. Meanwhile, by Taylor expansion, for any $t \in \R$, we have 
\begin{equation*}
U^N(t)=\frac{\sqrt{N}\left( \widetilde{F}(\xi^N(t))-p\right) }{\widetilde{f}(\widetilde{F}^{-1}(p))}
\to t, ~~~~\mbox{as}~~~ N \to \infty.
\end{equation*}
Set $\kappa^N=\sqrt{N}(\widetilde{ F^N}^{-1}(p) -\widetilde{F}^{-1}(p))$. 
We first show that $\kappa^N$ is bounded in probability.
For any $t \in \R$, we have 
\begin{align*}
\p(\kappa^N \le t)
&=\p(p\le \widetilde{F^N}(\xi^N(t)))\\
&=\p\left( \frac{1}{N}\sum_{k=1}^N\left((2p-1)\left( \xi^N(t)-X^{(k)}\right)\id_{\{X^{(k)}\le \xi^N(t)\}} +pX^{(k)} \right)- p\xi^N(t)\le  0\right)\\
&=\p\left( \frac{1}{N}\sum_{k=1}^N(1-p)\left( X^{(k)}-\xi^N(t)\right)\id_{\{X^{(k)}\le \xi^N(t)\}} +p\left( X^{(k)}-\xi^N(t)\right)\id_{\{X^{(k)}>\xi^N(t)\}}\le  0\right)\\
&\le \p\left( \frac{1}{N}\sum_{k=1}^N(1-p)\left( X^{(k)}-\xi^1(t)\right)\id_{\{X^{(k)}\le \xi^1(t)\}} +p\left( X^{(k)}-\xi^1(t)\right)\id_{\{X^{(k)}>\xi^1(t)\}}\le  0\right).
\end{align*}
By CLT, we have 
$$\frac{1}{N}\sum_{k=1}^N(1-p)\left( X^{(k)}-\xi^1(t)\right)\id_{\{X^{(k)}\le \xi^1(t)\}} +p\left( X^{(k)}-\xi^1(t)\right)\id_{\{X^{(k)}>\xi^1(t)\}} \sim \mathrm{N}(\mu(t), \sigma(t)^2)$$
with $\mu(t)=(1-p)\E\left[\left( X^{(k)}-\xi^1(t)\right)_-\right] +p\E\left[\left( X^{(k)}-\xi^1(t)\right)_+\right]$ and $\sigma(t)<\infty$. It is clear that $\mu(t) \to +\infty$ as $t \to -\infty$, which gives $\p(\kappa^N \le t)  \to 0$ as $t \to -\infty$. Similarly, we have $\p(\kappa^N > t)  \to 0$ as $t \to \infty$.
Hence, for any $\epsilon>0$, there exists a constant $t^*$ such that $\p(\vert \kappa^N\vert >t^*)\le \epsilon$ for all $N$.

Moreover,  for any $\varepsilon>0$ and $t \in \R$, we have 
$$
\mathbb{P}\left( \kappa^N\leq t,\;Z^N(0)\geq t+\varepsilon\right) =\mathbb{P}\left( Z^N(t)\leq U^N(t),\;Z^N(0)\geq t+\varepsilon\right)\to0,~~~ \mbox{as}~~~ N \to \infty.
$$
Similarly, we can get $\mathbb{P}\left( \kappa^N\geq t,\;Z^N(0)\leq t-\varepsilon\right)\to 0$.

Let $a_i=-t^*+2ti/M$ for $i \in \{0, 1, \dots, M\}$. Then,  
\begin{align*}
\p\left(\vert \kappa^N-Z^N(0)\vert >\epsilon\right)&\le \p\left(\vert\kappa^N\vert>t^*\right)+\p\left(\vert \kappa^N\vert \le M,\vert \kappa^N-Z^N(0)\vert >\epsilon \right)\\
&\le \epsilon+\sum_{i=1}^M \p\left(a_{i-1}\le \kappa^N\le a_i, \vert \kappa^N-Z^N(0)\vert >\epsilon \right)\\
&\le \epsilon+\sum_{i=1}^M \p\left( \kappa^N\le a_i, Z^N(0)\ge a_i+\epsilon/2 \right)+\p\left( \kappa^N\ge a_{i-1}, Z^N(0)\le a_{i-1}-\epsilon/2 \right).
\end{align*}
This implies that  $\kappa^N-Z^N(0)= \smallO_\p(1)$, i.e.,  $$
\widetilde{  F^N}^{-1}(p)=\widetilde{F}^{-1}(p)+\frac{p-\widetilde{ F^N}(\widetilde{F}^{-1}(p))}{\widetilde{f}(\widetilde{F}^{-1}(p))}+\smallO_\p\left(\frac{1}{\sqrt{N}}\right),~~~~p\in(0,1),
$$ 
which completes the proof. 
\end{proof}

The following theorem characterizes the asymptotic normality of $\DQex$. 
\begin{theorem}\label{thm:6}
Let $\alpha \in (0,1/2)$ and $\bX^{(1)}, \bX^{(2)}, \dots$ are i.i.d.~samples of $\bX=(X_1, \dots, X_n) \in \X^n$. Let $t=\sum_{i=1}^n \ex_\alpha(X_i)$, $\theta_1=\E[(S-t)_+]$ and $\theta_2=\E[\vert S-t\vert ]$ where $S=\sum_{i=1}^n X_i$.  Let $y_i=\widetilde{F}_i^{-1}(1-\alpha)$, $\mu_i=\mathbb{E}[X_i]$ and $\mu_i^{-}=\mathbb{E}[(X_i-y_i)_{-}]$ for $i\in[n]$. As $N \to \infty$, we have 
$$\sqrt{N} \left(\wDQex(N)-\DQex(\bX)\right) \lawto \mathrm{N}(0, \sigma^2_{\ex})$$
with  $ \sigma^2_\ex= \mathbf A^\top_\ex \Sigma_\ex\mathbf A_\ex$, where $\mathbf A_{\ex}=  \nabla h(\mathbf{v})\nabla g(\boldsymbol{\theta})^\top,$ and $\Sigma_\ex$ is the covariance matrix of random vector 
$$\left((S-t)_+, |S-t|, (y_1-X_1)\mathbb{I}_{\{X_1\leq y_1\}}, X_1,  \dots, (y_n-X_n)\mathbb{I}_{\{X_n\leq y_n\}}, X_n\right)^\top.$$ 
 Further, $\nabla h(\mathbf{v})$ is given by $$
\nabla h(\mathbf{v})^\top = 
\begin{pmatrix}
1 & 0 & \dfrac{1 - G(t)}{\widetilde{f}_1(y_1)} \cdot \dfrac{D_1 - 2\mu_1^-}{D_1^2} & \dfrac{1 - G(t)}{\widetilde{f}_1(y_1)} \cdot \dfrac{-\mu_1^-}{D_1^2} & \cdots & \dfrac{1 - G(t)}{\widetilde{f}_n(y_n)} \cdot \dfrac{D_n - 2\mu_n^-}{D_n^2} & \dfrac{1 - G(t)}{\widetilde{f}_n(y_n)} \cdot \dfrac{-\mu_n^-}{D_n^2} \\
0 & 1 & \dfrac{1 - 2G(t)}{\widetilde{f}_1(y_1)} \cdot \dfrac{D_1 - 2\mu_1^-}{D_1^2} & \dfrac{1 - 2G(t)}{\widetilde{f}_1(y_1)} \cdot \dfrac{-\mu_1^-}{D_1^2} & \cdots & \dfrac{1 - 2G(t)}{\widetilde{f}_n(y_n)} \cdot \dfrac{D_n - 2\mu_n^-}{D_n^2} & \dfrac{1 - 2G(t)}{\widetilde{f}_n(y_n)} \cdot \dfrac{-\mu_n^-}{D_n^2}
\end{pmatrix}$$ with $
D_i = 2\mu_i^- + \mu_i - y_i, (i = 1, \dots, n)
,$  and  $ \nabla g(\boldsymbol{\theta})=\left(1 / (\alpha\theta_2),-\theta_1 /\left(\alpha\theta^2_2\right)\right).$ 
\end{theorem}
\begin{proof}
 By Theorem 1 of  \cite{HLWW24}, we have 
\begin{equation*}\label{eq:ex}
\DQex(\bX) =\frac{1}{\alpha}
\frac{\mathbb{E}\left[ (S-t) _+\right]}{\mathbb{E}\left[\left| S-t \right|\right]}.
\end{equation*}

Let $t^N=\sum_{i=1}^n \widehat{ \ex_{i,\alpha}}(N)$ where $\widehat {\ex_{i,\alpha}}(N)$ is the empirical estimator of $\ex_\alpha(X_i)$ for $i \in [n]$. Let $\theta^N_1=\sum_{k=1}^N (S^{(k)}-t^N)_+/N$ and $\theta^N_2=\sum_{k=1}^N \vert S^{(k)}-t^N\vert/N$.

For $\theta_1^N$, we have 
\begin{equation*}\label{eq:theta1}
   \theta_1^N-\theta_1=\frac{1}{N}\sum_{k=1}^{N}(S^{(k)}-t^N)_+-\frac{1}{N}\sum_{k=1}^{N}(S^{(k)}-t)_++\frac{1}{N}\sum_{k=1}^{N}(S^{(k)}-t)_+-\theta_1.
\end{equation*}
Moreover,
\begin{equation}\label{theta1}
   \begin{aligned}
    \frac{1}{N}\sum_{k=1}^{N}(S^{(k)}-t^N)_+-\frac{1}{N}\sum_{k=1}^{N}(S^{(k)}-t)_+&=\int_{-\infty}^{\infty}(x-t^N)_+\d \widehat G^N(x)-\int_\infty^{\infty}(x-t)_+\d \widehat G^N(x)\\
&=\int_{t^N}^{\infty}(x-t^N)\d \widehat G^N(x)-\int_t^{\infty}(x-t)\d \widehat G^N(x)\\
    &=\int_{t^N}^t(x-t^N)\d \widehat G^N(x) +\int_t^{\infty}(t-t^N)\d \widehat G^N(x)\\
    &=\int_{t^N}^t(x-t^N)\d  \widehat G^N(x)+(t-t^N)(1-\widehat G^N(t)).
\end{aligned} 
\end{equation}
Notice that
$$\int_{t^N}^t(x-t^N)\d \widehat G^N(x)\leq |(t-t^N)(\widehat G^N(t)-\widehat G^N(t^N))|=\smallO_\p(N^{-\frac{1}{2}}).$$
So \eqref{theta1} can be written into 
\begin{equation*}
  \frac{1}{N}\sum_{k=1}^{N}(S^{(k)}-t^N)_+-\frac{1}{N}\sum_{k=1}^{N}(S^{(k)}-t)_+=(t-t^N)(1-G(t))+\smallO_\p(N^{-\frac{1}{2}}). 
\end{equation*}
Then we have
\begin{equation*}
    \theta_1^N-\theta_1=\frac{1}{N}\sum_{k=1}^{N}(S^{(k)}-t)_+-\theta_1+(t-t^N)(1-G(t))+\smallO_\p(N^{-\frac{1}{2}}). 
\end{equation*}
Similarly, for $\theta_2$ we can get
\begin{equation*}
    \theta_2^N-\theta_2=\frac{1}{N}\sum_{k=1}^{N}|S^{(k)}-t|-\theta_2+(t-t^N)(1-2G(t))+\smallO_\p(N^{-\frac{1}{2}}).
\end{equation*}
By Lemma \ref{lemma:ex}, we obtain
\begin{equation}\label{eq:tN}
\begin{aligned}
    t-t^N&=\sum_{i=1}^{n}\left(\widetilde{F}_i^{-1}(1-\alpha)-\widetilde{F_i^N}^{-1}(1-\alpha)\right)\\
&=\sum_{i=1}^n\frac{\widetilde{F^N_i}(\widetilde{F}_i^{-1}(1-\alpha))-(1-\alpha)}{\widetilde{f}_i(\widetilde{F}_i^{-1}(1-\alpha))}+\smallO_\p(N^{-\frac{1}{2}}).
\end{aligned}
\end{equation}
For simplicity, we denote $y_i=\widetilde{F}_i^{-1}(1-\alpha)$. With this notation, together with \eqref{eq:tildeF}, we have 
\begin{equation}
\begin{aligned}
\widetilde{F_i^N}\left(y_i\right) & =\frac{y_i F_i^N\left(y_i\right)-\int_{-\infty}^{y_i} x \mathrm{~d} F_i^N(x)}{2\left(y_i F_i^N\left(y_i\right)-\int_{-\infty}^{y_i} x \mathrm{~d} \widehat{F}_i(x)\right)+\widehat{x}_i^N-y_i} \\
& =\frac{\frac{1}{N} \sum_{k=1}^N\left(y_i-X_i^{(k)}\right) \mathbb{I}_{\left\{X_i^{(k)} \leqslant y_i\right\}}}{2\left(\frac{1}{N} \sum_{k=1}^N\left(y_i-X_i^{(k)}\right) \mathbb{I}_{\left\{X_i^{(k)} \leqslant y_i\right\}}\right)+\frac{1}{N} \sum_{k=1}^N X_i^{(k)}-y_i},
\end{aligned}\label{FiN}
\end{equation}
where $\widehat x_i^N$ is the empirical estimators of $\E[X_i].$ 

Define the function $h:\mathbb{R}^{2n+2}\to\mathbb{R}^2$ by
$$
  h(\mathbf{x}):=(h_1(\mathbf{x}),h_2(\mathbf{x}))=
  \begin{pmatrix}
    
      x_1+\left(1-G(t)\right)\displaystyle\sum_{i=1}^n\frac{1}{\widetilde{f}_i(y_i)}\frac{x_{2i+1}}{2x_{2i+1}+x_{2i+2}-y_i}\\ x_2+(1-2G(t))\displaystyle\sum_{i=1}^n\frac{1}{\widetilde{f}_i(y_i)}\frac{x_{2i+1}}{2x_{2i+1}+x_{2i+2}-y_i},
  \end{pmatrix}
$$
where $\mathbf{x}=(x_1, x_2, \dots, x_{2n+2})$. 
Let $\mu_i=\mathbb{E}[X_i]$ and $\mu_i^{-}=\mathbb{E}[(X_i-y_i)_{-}]$.    Define $\mathbf V\in\mathbb{R}^{2n+2}$  as 
\begin{equation*}
    \mathbf{V}=\left((S-t)_+, |S-t|, (y_1-X_1)\mathbb{I}_{\{X_1\leq y_1\}}, X_1,  \dots,(y_n-X_n)\mathbb{I}_{\{X_n\leq y_n\}}, X_n\right)^\top.
\end{equation*}
Suppose $\mathbf {V}^{(1)}$, $\mathbf {V}^{(2)},...$ are i.i.d.~samples of $\mathbf V$, and let $\overline{\mathbf V}^N= \sum_{i=1}^N{\mathbf V}^{(k)}/N$. By multivariate central limit theorem, $\sqrt{N}\overline{V}^N\stackrel{d}{\longrightarrow}\mathrm{N}(\mathbf v, \Sigma_{\ex})$, where $\mathbf v=(\theta_1, \theta_2, \mu_1^-,\mu_1,\dots, \mu_n^-,\mu_n)^\top$, and $\Sigma_\ex$ is covariance matrix of $\mathbf{V}$.
 By Delta method, we have
\begin{equation}
    \sqrt{N}(\theta_1^N-\theta_1, \theta_2^N-\theta_2)^\top=\sqrt{N}(h(\overline{
\mathbf V}^N)-h(\mathbf v))+\smallO_\p(1)\stackrel{d}{\to}\mathrm{N}(0, \mathbf \nabla h(\mathbf{v}) ^\top\Sigma_\ex \nabla h(\mathbf{v})  )).
\end{equation}

Take $g\left(x_1, x_2\right)={x_1}/{ \alpha x_2}$, then we have   
 $ \nabla g(\boldsymbol{\theta} )=\left(1 /  \alpha \theta_2,-\theta_1 /\left( \alpha\theta^2_2\right)\right)$ with $\boldsymbol{\theta}=(\theta_1,\theta_2)$.   
 Note that 
\begin{equation*}
\begin{aligned}
    &g\left(h(\overline{
\mathbf V}^N)-h(\mathbf v)+\smallO_\p(\frac{1}{\sqrt{N}})\right)
\\
    &=\frac{1}{\alpha}\frac{h_1(\overline{ \mathbf V}^N)-h_1(\mathbf v)+\smallO_\p(\frac{1}{\sqrt{N}})}{h_2(\overline{ \mathbf V}^N)-h_2(\mathbf v)+\smallO_\p(\frac{1}{\sqrt{N}})}-\frac{1}{\alpha}\frac{h_1(\overline{ \mathbf V}^N)-h_1(\mathbf v)}{h_2(\overline{ \mathbf V}^N)-h_2(\mathbf v)}+\frac{1}{\alpha}\frac{h_1(\overline{ \mathbf V}^N)-h_1(\mathbf v)}{h_2(\overline{ \mathbf V}^N)-h_2(\mathbf v)}\\
    &=\frac{1}{\alpha}\frac{\left(h_2(\overline{
\mathbf V}^N)-h_2(\mathbf v)-h_1(\overline{
\mathbf V}^N)+h_1(\mathbf v)\right)\smallO_\p(\frac{1}{\sqrt{N}})}{\left(h_2(\overline{
\mathbf V}^N)-h_2(\mathbf v)+\smallO_\p(\frac{1}{\sqrt{N}})\right)(h_2(\overline{
\mathbf V}^N)-h_2(\mathbf v))}+\frac{1}{\alpha}\frac{h_1(\overline{
\mathbf V}^N)-h_1(\mathbf v)}{h_2(\overline{
\mathbf V}^N)-h_2(\mathbf v)}\\
    &=\frac{1}{\alpha}\frac{h_1(\overline{
\mathbf V}^N)-h_1(\mathbf v)}{h_2(\overline{
\mathbf V}^N)-h_2(\mathbf v)}+\smallO_\p(\frac{1}{\sqrt{N}}).
\end{aligned}    
\end{equation*}
  Therefore, we have $\sqrt{N}\left(g(\mathbf{\theta}^N)-g(\mathbf{\theta})\right)=\sqrt{N}\left(g(h(\overline{
\mathbf V}^N)-h(\mathbf v))\right)+\smallO_\p(1)$.  Define  $\mathbf A_{\ex}=  \nabla h(\mathbf{v})\nabla g(\mathbf \theta) $.  Using the  Delta method again,  we have 
\begin{equation*}
    \sqrt{N}(\widehat{\mathrm{DQ}^{\ex}_{\alpha}}(N)-\DQex(\mathbf{X}))=\frac{\sqrt{N}}{\alpha }\left(\frac{\theta_1^N}{\theta_2^N}-\frac{\theta_1}{\theta_2}\right)\stackrel{d}{\to}\mathrm{N}(0,\sigma^2_\ex),
\end{equation*}
where $ \sigma^2_\ex= \mathbf A^\top_\ex \Sigma_\ex\mathbf A_\ex$. 
\end{proof}

In the following analysis, we adopt the same model framework as in Section \ref{sec:7}, assuming that $\mathbf{X} \sim \mathrm{N}(\boldsymbol{\mu}, \Sigma)$ and $\mathbf{Y} \sim t(\nu, \boldsymbol{\mu}, \Sigma)$, where the covariance matrix $\Sigma$ is specified in equation~\eqref{eq:matrix}. The default parameter settings are $r = 0.3$, $n = 5$, $\nu = 3$, and $\alpha = 0.1$.

Figure~\ref{fig:DQ_ex} illustrates that the empirical estimates of $\mathrm{DQ}^{\ex}_\alpha$ for $\mathbf{X} \sim \mathrm{N}(\boldsymbol{\mu}, \Sigma)$ closely follow the normal distribution $\mathrm{N}(0.33,0.78/N)$. Similarly, for $\mathbf{Y} \sim t(\nu, \boldsymbol{\mu}, \Sigma)$, the estimates align well with the distribution $\mathrm{N}(0.45,2.53/N)$.

\begin{figure}[htb!]
\centering
 \includegraphics[width=16cm]{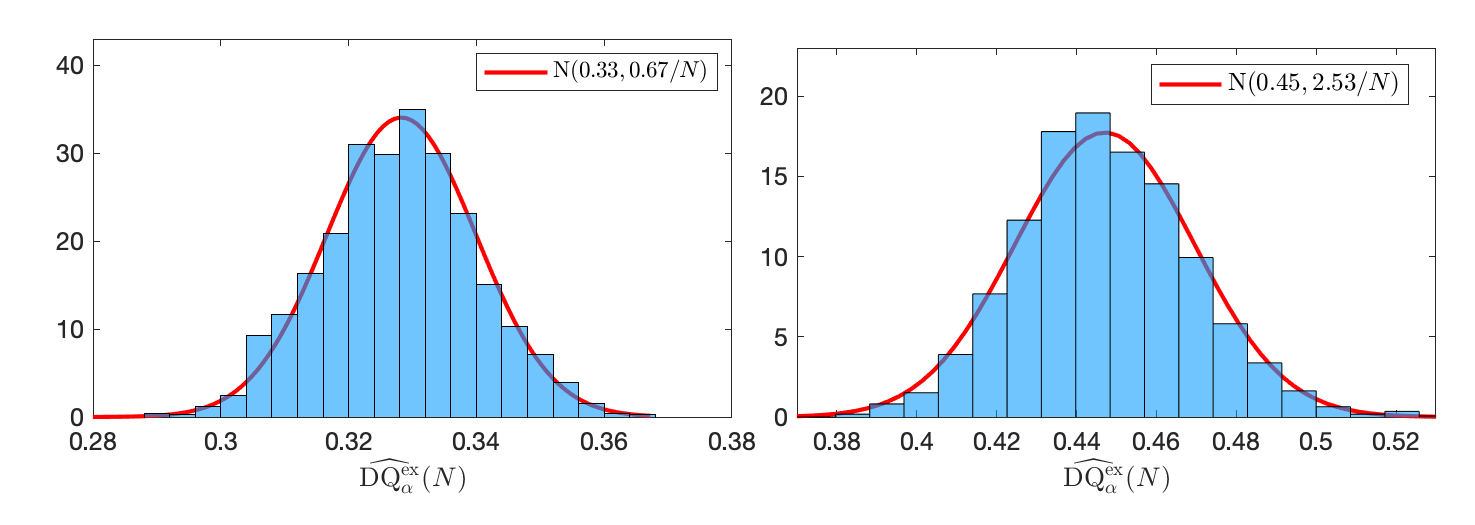}
 \captionsetup{font=small}
{\caption{   \small   Histogram of $\widehat{\mathrm{DQ}^{\ex}_\alpha}(N)$  for $\mathbf X\sim  \mathrm N(\boldsymbol{\mu},\Sigma)$ (left panel) and $ \mathbf{Y}\sim \mathrm t(\nu,\boldsymbol{\mu},\Sigma)$ (right panel)}\label{fig:DQ_ex}}
\end{figure}

We also examine how the asymptotic variance $\sigma^2_{\ex}$, as derived in Theorem~\ref{thm:6}, varies with respect to $\alpha$, $r$, $n$, and $\nu$ in Figures~\ref{sigma_ex_alpha} and~\ref{sigma_ex_nu}. Since the resulting curves exhibit similar patterns to those observed for $\mathrm{DQ}_\alpha^{\mathrm{VaR}}$ and $\mathrm{DQ}_\alpha^{\mathrm{ES}}$ in Section \ref{sec:7}, we omit further discussion.

\begin{figure}[htb!]
\centering
\includegraphics[width=16cm]{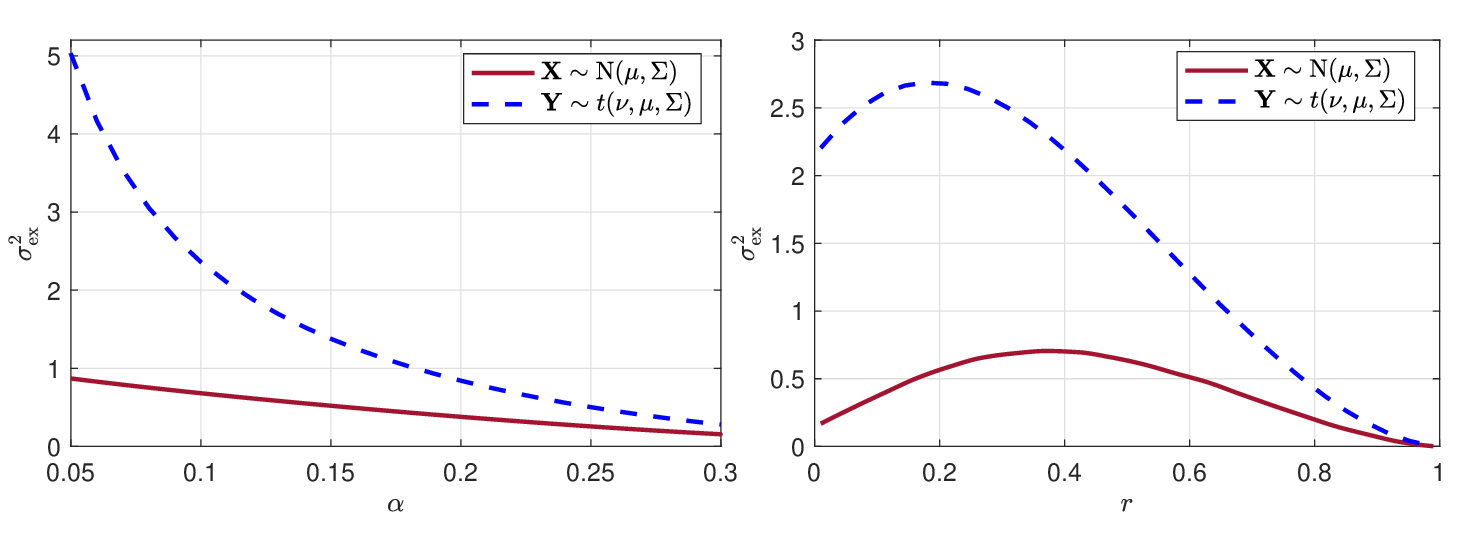}
 \captionsetup{font=small}
{\caption{   \small   Asymptotic variance    for $\alpha\in[0.05,0.3]$ with fixed $n=5$, $\nu=3$ and $r=0.3$ (left),  and for  $r\in[0.01,0.99]$ with fixed $n=5$, $\alpha=0.1$ and $\nu=3$  (right)}\label{sigma_ex_alpha}}
\end{figure}

\begin{figure}[htb!]
\centering
\includegraphics[width=16cm]{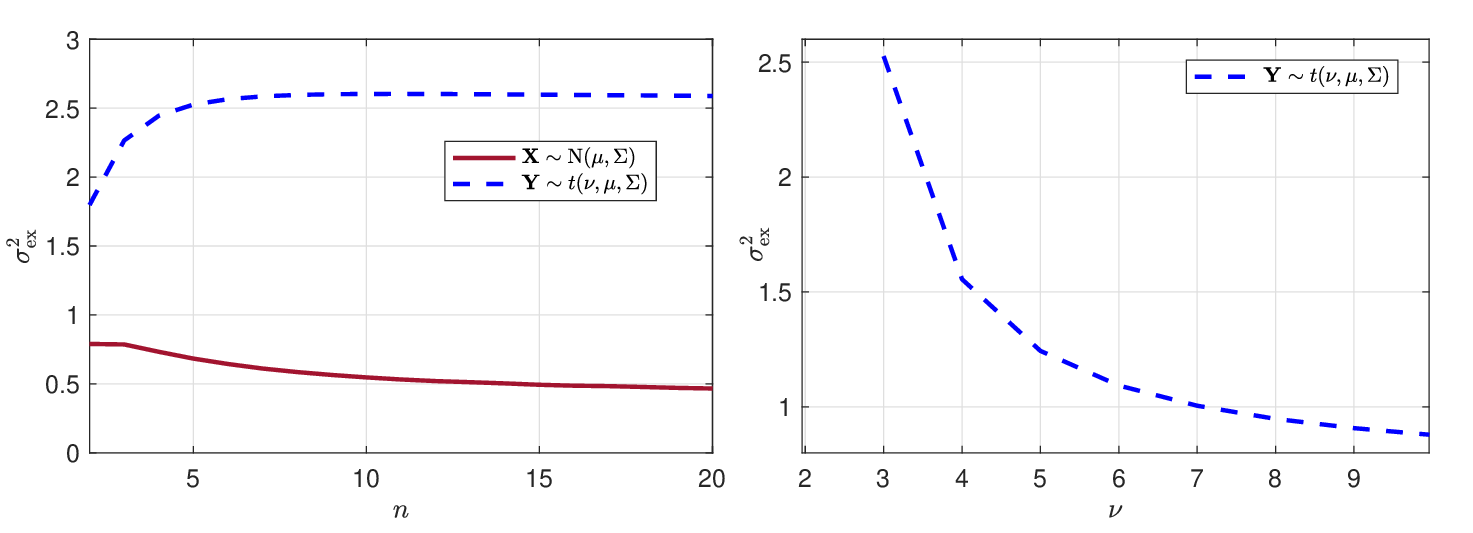}
 \captionsetup{font=small}
{\caption{ \small   Asymptotic variance    for $n\in[2,20]$ with fixed $\alpha=0.1$, $\nu=3$ and $r=0.3$ (left),  and for  $\nu\in[3,10]$ with fixed $n=5$, $\alpha=0.1$ and $r=0.3$  (right)}\label{sigma_ex_nu}}
\end{figure}

 Finally, using the financial data described in Section \ref{sec:real_data}, we construct confidence intervals for $\DQex$ in Figures \ref{fig:DQ_ex_0.05} and \ref{fig:DQ_ex_0.1}. We find that the pattern of $\DQex$ is similar to those of $\DQVaR$ and $\DQES$, but its confidence intervals are noticeably narrower. In particular, the variance of $\DQex$ is 0.0039 for $\alpha=0.05$ and 0.0026 for $\alpha=0.1$, which are much smaller than the corresponding variances of $\DQVaR$ and $\DQES$. This observation is consistent with the results under elliptical models and highlights the advantage noted in \cite{HLWW24} that DQ based on expectiles can handle the challenges of computation with small samples, which pose practical difficulties for VaR, ES, and their corresponding DQ due to the scarcity of tail data.

\begin{figure}[htb!]
\centering
 \includegraphics[width=14cm]{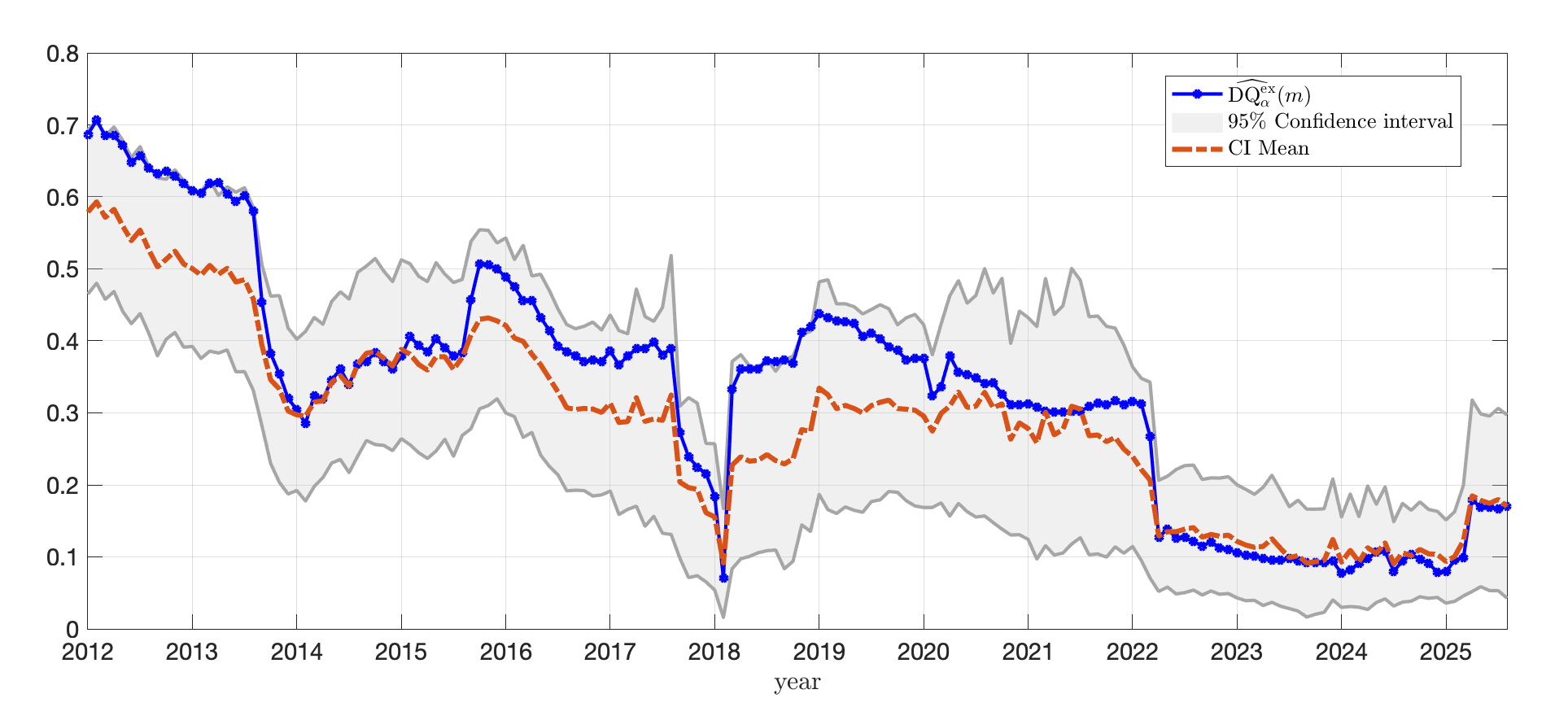}
 \captionsetup{font=small}
{\caption{   \small    Bootstrap confidence bands for $\DQex$ with $\alpha=0.05$}\label{fig:DQ_ex_0.05}}
\end{figure}

\begin{figure}[htb!]
\centering
 \includegraphics[width=14cm]{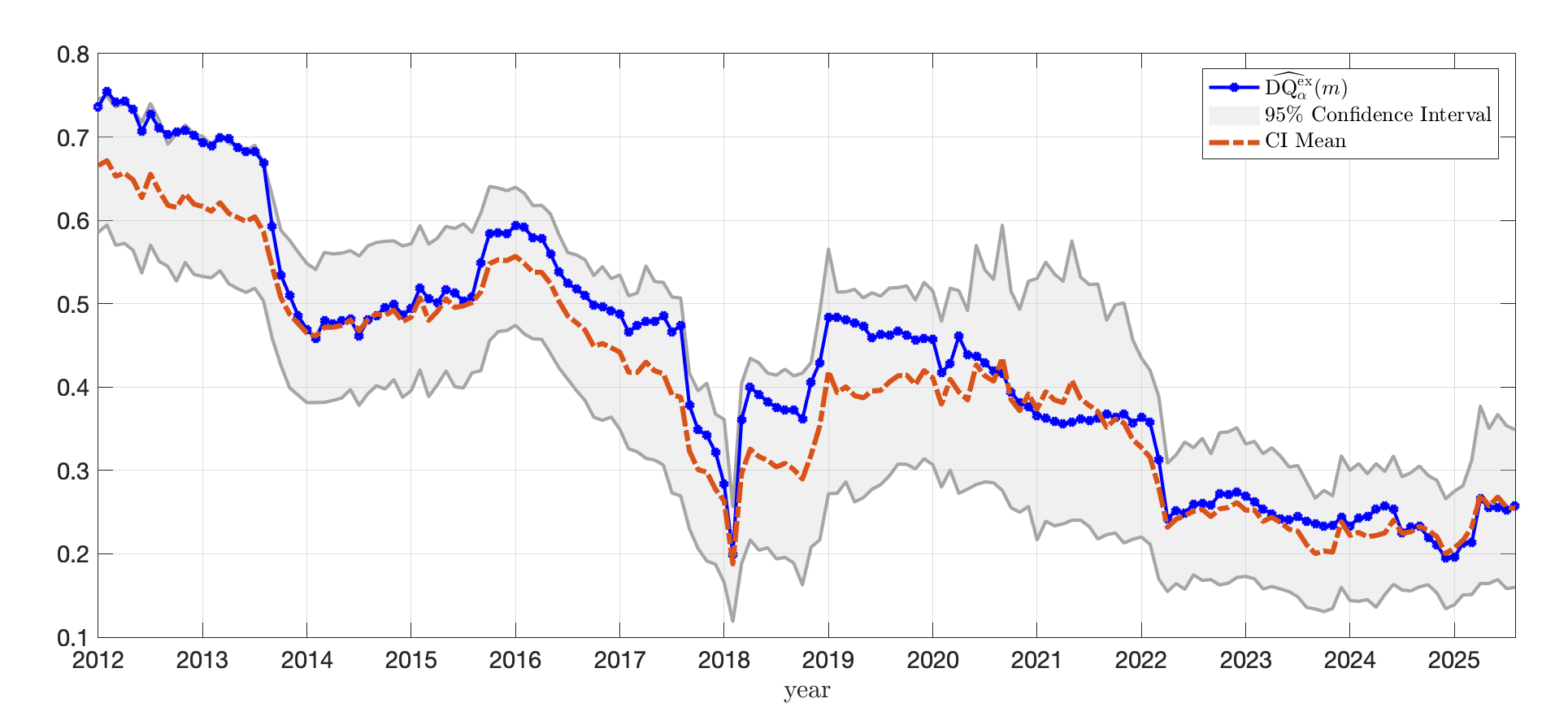}
 \captionsetup{font=small}
{\caption{   \small   Bootstrap confidence bands for $\DQex$ with $\alpha=0.1$}\label{fig:DQ_ex_0.1}}
\end{figure}

\section{The proof of the main context}
\label{App:B}
\subsection{Proof in Section \ref{sec:3}}

\begin{proof}[Proof of Proposition \ref{prop:1}]
The proof is straightforward for the fact that $X_i \laweq Y_i$, $i\in [n]$ and $\sum_{i=1}^n X_i \laweq \sum_{i=1}^n Y_i$ for any $\mathbf X=(X_1, \dots, X_n), \mathbf Y=(Y_1, \dots, Y_n) \in \X^n$.
\end{proof}

\begin{proof}[Proof of Theorem \ref{thm:1}]
Let $\bX=(X_1, \dots, X_n)$ and  $\bX^{(k)}=(X^{(k)}_1, \dots, X^{(k)}_n)$ for $k \in \N$.
Define the function $h: I \to \R$ and $h^k: I \to \R$ for every $k \in \N$  as 
 $$h(\beta)=\rho_{\beta}\left(\sum_{i=1}^n X_i\right)-\sum_{i=1}^{n}\rho_{\alpha}(X_i)~~~\mbox{and} ~~~h^k(\beta)=\rho_{\beta}\left(\sum_{i=1}^n X^{(k)}_i\right)-\sum_{i=1}^{n}\rho_{\alpha}(X^{(k)}_i),~~~\beta\in I.$$ 
Since $\bX^{(k)} \lawto \bX$, we have $X^{(k)}_i \lawto X_i$ for all $i\in [n]$ and $\sum_{i=1}^n X^{(k)}_i \lawto \sum_{i=1}^n X_i$.
If (a) or (b) holds, we have $\rho_\alpha(X^{(k)}_i) \to \rho_\alpha(X_i)$ for all $i\in [n]$ and $\rho_\beta(\sum_{i=1}^n X^{(k)}_i) \to \rho_\beta(\sum_{i=1}^n X_i)$  for all $\beta \in I$ as $k \to \infty$. Hence, $h^k(\beta) \to h(\beta)$ as $k \to \infty$ for all $\beta \in I$.
 Moreover, it is clear that $h$ and $h^k$, $k \in \N$, are decreasing functions. 

 If $\DQrho(\bX)=0$ or $\widetilde \alpha=\sup I$, the conclusion holds trivially. Without loss of generality, we assume that $\DQrho(\bX)>0$ and $\widetilde \alpha<\sup I$.
For a given  $\beta< \alpha^*=\inf\{\beta \in I: h(\beta)\le 0\}$, it is clear that $h(\beta)>0$. Hence, there exists $K_\beta$ such that $h^k(\beta)>0$ for all $k>K_\beta$. Since $h^k$, $k \in \N$,  are  decreasing functions, we have $\alpha^k=\inf\{\beta \in I: h^k(\beta)\le 0\}\ge \beta$ for all $k >K_\beta$, which implies $\liminf \alpha^k\ge \alpha^*$. By Definition \ref{def:DQ}, we have $\liminf_{k\to\infty}\DQrho(\mathbf X^{(k)})\ge \DQrho(\mathbf X)$.

For a given  $\beta> \widetilde\alpha=\sup\{\beta \in I: h(\beta)\ge 0\}$, we have  $h(\beta)<0$. Hence, there exists $K_\beta$ such that $h^k(\beta)<0$ for all $k>K_\beta$. Therefore, $\widetilde\alpha^k=\sup\{\beta \in I: h^k(\beta)\ge 0\}\le \beta$, which implies $\limsup_{k \to \infty} \widetilde\alpha^k\le \widetilde\alpha$. Since $\alpha^k \le \widetilde \alpha^k$, by Definition \ref{def:DQ}, we have $\limsup_{k\to\infty}\DQrho(\bX^{(k)})\le {\widetilde{\alpha}
}/{\alpha}$.

Next, assume $\rho_\beta(\sum_{i=1}^n X_i)$ is a strictly decreasing function.
In this case, we have $\alpha^*=\widetilde \alpha$. Therefore,
$$\frac{\alpha^*}{\alpha}=\liminf_{k\to\infty}\DQrho(\mathbf X^{(k)})\le  \limsup_{k\to\infty}\DQrho(\mathbf X^{(k)}) \le \frac{\widetilde\alpha}{\alpha},$$
which gives $\liminf_{k\to\infty}\DQrho(\bX^{(k)})=  \limsup_{k\to\infty}\DQrho(\bX^{(k)})=\DQrho(\bX)$.
\end{proof}

\subsection{Proofs in Section \ref{sec:4}}\label{app:2}

To show the asymptotical normality,  we first present two lemmas that will be helpful in the discussion of asymptotic distribution for DQ's empirical estimators.
A Gaussian process $\{B(t); t \in [0,1]\}$ is called Brownian bridge if $\E[B(t)]=0$ for all $t \in [0,1]$ and $\E[B(t_1)B(t_2)]=t_1 \wedge t_2- t_1t_2$ for any $t_1, t_2 \in [0,1]$.
\begin{lemma}[Theorem 3.1.1 in \citet{CM83}]\label{lem:1}
Let $U_1, U_2, \dots$ be a sequence of independently uniformly distributed random variables on $[0,1]$. For $N \in \mathbb N$ and $y \in \R$, let $E^N(y)=\sum_{k=1}^N\mathbb{I}_{\{U_k\leq y\}}/N$ and $\delta^N(y)=\sqrt{N}(E^N(y)-y)$.
There exists a probability space on which one can define a Brownian bridge $\{B^N(y); 0\le y\le 1\}$ for each $N$ such that 
$$
\sup_{0\le y \le 1}\vert \delta^N(y)-B^N(y)\vert =  \O(N^{-1/2}\log N)~~~ \mbox{a.s.}
$$   
\end{lemma}
\begin{lemma}[\cite{G71}]\label{lem:Bahadur}
Let $\alpha \in (0,1)$. Suppose $F'(F^{-1}(\alpha))=f(F^{-1}(\alpha))>0$. For empirical distribution $\wF^N$ built from i.i.d.~samples, we have 
$$(\wF^N)^{-1}(\alpha)=F^{-1}(\alpha)+\frac{\alpha-\wF^N(F^{-1}(\alpha))}{f(F^{-1}(\alpha))}+\smallO_\p(N^{-1/2}).$$
\end{lemma}

\begin{proof}[Proof of Theorem \ref{thm:2}]
   By Proposition \ref{prop:3}, we know that 
$$\wDQVaR(N)-\mathrm{DQ}^{\VaR}_{\alpha}(\mathbf{X})=\frac{1}{\alpha}\left( {\wG^N}\left(\sum_{i=1}^{n}(\wF^N_i)^{-1}(1-\alpha)\right)-G\left(\sum_{i=1}^{n}F_i^{-1}(1-\alpha)\right)\right). 
$$
Using a straightforward decomposition, we have 
\begin{equation*}
   \begin{aligned}
&{\wG^N}\left(\sum_{i=1}^{n}(\wF^N_i)^{-1}(1-\alpha)\right)-G\left(\sum_{i=1}^{n}F_i^{-1}(1-\alpha)\right)\\=~&\underbrace{{\wG^N}\left(\sum_{i=1}^{n}(\wF^N_i)^{-1}(1-\alpha)\right)-G\left(\sum_{i=1}^{n}(\wF^N_i)^{-1}(1-\alpha)\right)}_{I_1}+\underbrace{G\left(\sum_{i=1}^{n}(\wF^N_i)^{-1}(1-\alpha)\right)-G\left(\sum_{i=1}^{n}F_i^{-1}(1-\alpha)\right)}_{I_2}\\=~&I_1+I_2.
\end{aligned} 
\end{equation*}


We first consider $I_1$.
For any $N \in \N$ and $y \in \R$, let $\gamma^N(y)=\sqrt{N}(\wG^N(y)-G(y))$.
Since $\sum_{i=1}^n X^{(k)}_i$, $k \in \N$, are i.i.d.,  there exist independently uniformly distributed random variables $U^{(1)}, U^{(2)}, \dots$, such that $\sum_{i=1}^n X^{(k)}_i=G^{-1}(U^\k)$ for $k \in \N$. Let $E^N(y)=\sum_{k=1}^N\mathbb{I}_{\{U^{(k)}\leq y\}}/N$ and $\delta^N(y)=\sqrt{N}(E^N(y)-y)$. Hence,  $\gamma^N(y)=\delta^N(G(y))$.
By Lemma \ref{lem:1}, we have 
\begin{align*}
 \sup_{y\in \mathbb{R}}|\gamma^N(y)-B^N(G(y))|= \sup_{y\in \mathbb{R}}|\delta^N(G(y))-B^N(G(y))|=\sup_{0\leq y\leq1}|\delta^N(y)-B^N(y)|=\smallO_\p(1),
\end{align*}
where $B^N$ is a standard Brownian bridge for each $N \in \N$.
 Since $\sqrt{N}I_1=\gamma^N\left(\sum_{i=1}^{n}(\wF^N_i)^{-1}(1-\alpha)\right)$, we have
 $$\left\vert \sqrt{N}I_1-B^N\left(G\left(\sum_{i=1}^{n}(\wF^N_i)^{-1}(1-\alpha)\right)\right)\right\vert\le  \sup_{y\in \mathbb{R}}|\gamma^N(y)-B^N(G(y))|=\smallO_\p(1),$$
 which implies
\begin{equation}\label{eq:DQVaR1}
    \sqrt{N}I_1=B^N\left(G\left(\sum_{i=1}^{n}(\wF^N_i)^{-1}(1-\alpha)\right)\right)+\smallO_\p(1).
\end{equation}
Using the same argument, we can show that 
\begin{equation}\label{eq:DQVaR2}
\sqrt{N}\left(  \wG^N\left( \sum_{i=1}^{n}F_i^{-1}(1-\alpha)\right)-G\left( \sum_{i=1}^{n}F_i^{-1}(1-\alpha)\right) \right)=B^N\left(G\left(\sum_{i=1}^{n}F_i^{-1}(1-\alpha)\right)\right)+\smallO_\p(1).
\end{equation}

For each standard Brownian bridge $B^N$, $N \in \N$, by \citet[Theorem 1.5.2]{CM83}, we have 
$$
\lim_{h\to 0} \sup_{0\leq t\leq 1-h}\frac{|B^N(t+h)-B^N(t)|}{(2h\ln (1/h))^{1/2}}=1 ~~~\mbox{a.s.}$$
Together with the fact that  $\sum_{i=1}^{n}(\wF^N_i)^{-1}(1-\alpha) \to \sum_{i=1}^{n}F_i^{-1}(1-\alpha)$ a.s. and   $G$ is continuous and its density function is  continuous and   bounded density at $\sum_{i=1}^n F_i^{-1}(1-\alpha)$ in Assumption \ref{ass:DQVaR}, we have 
\begin{equation}\label{eq:DQVaR3}
B^N\left(G\left(\sum_{i=1}^{n}(\wF^N_i)^{-1}(1-\alpha)\right)\right)=B^N\left(G\left(\sum_{i=1}^{n}F_i^{-1}(1-\alpha)\right)\right)+\smallO_\p(1).
\end{equation}
Combine \eqref{eq:DQVaR1} -- \eqref{eq:DQVaR3} and $t_{n+1}=\sum_{i=1}^n F^{-1}_i(1-\alpha)$, we further have 
\begin{equation}\label{eq:I_1}
\sqrt{N}I_1=\sqrt{N}\left(  \wG^N\left( t_{n+1}\right)-G\left( t_{n+1}\right) \right)+\smallO_\p(1).
\end{equation}

 We now consider the term $I_2$ by employing Lemma \ref{lem:Bahadur} along with the mean value theorem. Recall that $t_i=F_i^{-1}(1-\alpha)$ for $i \in [n]$ and $t_{n+1}=\sum_{i=1}^n t_i$. By Assumption \ref{ass:DQVaR},  we have  

\begin{align*}
\sqrt{N}I_2&=\sqrt{N}\left(G\left(\sum_{i=1}^{n}(\wF^N_i)^{-1}(1-\alpha)\right)-G\left(\sum_{i=1}^{n}F_i^{-1}(1-\alpha)\right)\right) \\&=\sqrt{N}g\left(t_{n+1}\right)\left( \sum_{i=1}^{n}(\wF^N_i)^{-1}(1-\alpha)-\sum_{i=1}^{n}F_{i}^{-1}(1-\alpha)\right) +\smallO_\p(1)\\
&=\sqrt{N}g\left(t_{n+1}\right){\left( \sum_{i=1}^{n}\frac{F_i(t_i)-\wF^N_i(t_i)}{f_i(t_i)}\right)} +\smallO_\p(1).
\end{align*}
Let 
 $$\mathbf H=\left(\frac{g(t_{n+1})}{f_1(t_1)}, \dots, \frac{g(t_{n+1})}{f_n(t_n)}\right) ~~~\mbox{and}~~~{\mathbf M}^N=\sqrt{N}\left(F_1(t_1)-{\wF}^{N}_1(t_1),\dots,F_n(t_n))-{\wF}^{N}_n(t_n)\right) .
$$
We have 
\begin{equation}\label{eq:I_2}
\sqrt{N}I_2=    \mathbf H^\top\mathbf M^N+\smallO_\p(1).
\end{equation}

For $k \in \N$,  let  $\mathbf Z^\k=(Z^\k_1, \dots, Z^\k_{n+1})$ where $Z^\k_i=F_i(t_i)-\mathbb{I}_{\{X^\k_i\leq t_i\}}$ for $i\in [n]$ and $Z^\k_{n+1}=G(t_{n+1})-\mathbb{I}_{\{\sum_{i=1}^n X^\k_i\leq t_{n+1}\}}$.  
We have 
$$
\sqrt{N}\overline{\mathbf Z}^N:=\sqrt{N}\frac{\sum_{k=1}^N \mathbf Z^\k}{N}=\left( \mathbf M^N, \sqrt{N}(G(t_{n+1})-{\wG^N}(t_{n+1}))\right).
$$
Combining with \eqref{eq:I_1} and \eqref{eq:I_2}, we have
$$\sqrt{N}\left(\wDQVaR(N)-\mathrm{DQ}^{\VaR}_\alpha(\mathbf X)\right)=\sqrt{N}\frac{I_1+I_2}{\alpha}=\sqrt{N}\mathbf A_\VaR^\top\overline{\mathbf Z}^N+\smallO_\p(1),$$ where $\mathbf A_\VaR=\left(\mathbf H/\alpha, -1/\alpha\right)$.
It is clear that $\mathbf Z^\k$, $k\in \N$, are i.i.d.~random vectors.  By the multiple central limit theorem, we  obtain
$
\overline{\mathbf Z}^N\lawto {\rm N}(0,\Sigma_\VaR),
$
where $\Sigma_\VaR$ is the covariance matrix of $\mathbf Z^\k$. 
Therefore, $$\sqrt{N}\left(\wDQVaR(N)-\mathrm{DQ}^{\VaR}_\alpha(\mathbf X)\right)\lawto {\rm N}(0,\mathbf A_\VaR^\top\Sigma_\VaR \mathbf A_\VaR).$$
As $\bX^{(k)}\sim \bX$,   $\Sigma_\VaR$ is also the covariance matrix of the random vector $(\id_{\{X_1 \le t_1\}}, \dots, \id_{\{X_n \le t_n\}}, \id_{\{S\le t_{n+1}\}})$.
\end{proof}

\begin{proof}[Proof of Theorem \ref{thm:3}]
Assumption \ref{ass:unique} implies that $\alpha^*$ satisfies
\begin{equation}\label{eq:1}
\frac{1}{\alpha^*}\int_{1-\alpha^*}^1 G^{-1}(p) \d p=\sum_{i=1}^n \frac{1}{\alpha}\int_{1-\alpha}^1 F_i^{-1}(p)\d p.
\end{equation}
 Let $\alpha^N=\alpha\wDQES(N)$.  For the empirical version, we have
\begin{equation}\label{eq:2}
\frac{1}{\alpha^N}\int_{1-\alpha^N}^1 (\wG^N)^{-1}(p)\d p=\sum_{i=1}^n \frac{1}{\alpha}\int_{1-\alpha}^1 (\wF^N_i)^{-1}(p)\d p.
\end{equation}
Combine \eqref{eq:1} and \eqref{eq:2}, we get
\begin{align*}
0&=\underbrace{\frac{1}{\alpha^N}\int_{1-\alpha^N}^{1} (\wG^N)^{-1}(p)\d p-\frac{1}{\alpha^*}\int_{1-\alpha^*}^{1} (\wG^N)^{-1}(p)\d p}_{I_1}+\underbrace{\frac{1}{\alpha^*}\int_{1-\alpha^*}^{1} (\wG^N)^{-1}(p)\d p-\int_{1-\alpha^*}^1 G^{-1}(p) \d p}_{I_2}\\
&~~~-\underbrace{\left(\sum_{i=1}^n \frac{1}{\alpha}\int_{1-\alpha}^1  (\wF^N_i)^{-1}(p)\d p-\sum_{i=1}^n \frac{1}{\alpha}\int_{1-\alpha}^1 F_i^{-1}(p)\d p\right)}_{I_3}.
\end{align*}
For $I_1$,  the mean value theorem gives
$$
I_1=(\alpha^N-\alpha^*)\frac{ \widetilde{\alpha}^N{(\wG^N})^{-1}(1-\widetilde{\alpha}^N)-\int^{1}_{1-\widetilde{\alpha}^N}({\wG^N})^{-1}(p)\d p}{(\widetilde{\alpha}^N)^2},
$$
for some $\widetilde{\alpha}^N$  lies between $\alpha^*$ and $\alpha^N$. 
Since $S$ has a positive density function in Assumption \ref{ass:DQES}, $\beta \mapsto \ES_\beta(S)$ is a strictly decreasing function. 
By Proposition \ref{prop:consistency}, $\alpha^N$ is a strong consistent estimator of $\alpha^*$; that is $\alpha^N \to \alpha^*$ a.s. 
With the regularity and continuity conditions in Assumption \ref{ass:DQES}, we obtain
$$
\frac{ \widetilde{\alpha}^N{(\wG^N})^{-1}(1-\widetilde{\alpha}^N)-\int^{1}_{1-\widetilde{\alpha}^N}({\wG^N})^{-1}(p)\d p}{(\widetilde{\alpha}^N)^2}=\frac{ \alpha^*G^{-1}(1-\alpha^*)-\int^{1}_{1-\alpha^*}G^{-1}(p)\d p}{(\alpha^*)^2}+\smallO_\p(1).
$$

Let $S^{(k)}=\sum_{i=1}^n X^{(k)}_i$ for $k \in \N$. For $I_2$, by Lemma \ref{lem:Bahadur}, we have
\begin{equation}\label{eq:I_2_ES}\begin{aligned}
I_2&= \frac{1}{\alpha^*}\int_{1-\alpha^*}^{1} \left( (\wG^N)^{-1}(p)-G^{-1}(p)\right) \d p\\
&=\frac{1}{\alpha^*}\int_{1-\alpha^*}^{1} \frac{G(G^{-1}(p))-\wG^N(G^{-1}(p))}{g(G^{-1}(p))}\d p+\smallO_\p(1)\\
&=\frac{1}{\alpha^*}\int_{s}^{\infty} \left(G(x)-\wG^N(x)\right)\d x+\smallO_\p(1)\\
&=\frac{1}{\alpha^*}\left(\int_{s}^{\infty} \frac{1}{N}\sum_{k=1}^N\id_{\{S^{(k)}>x\}}\d x-\int_{s}^{\infty} 1-G(x)\d x\right)+\smallO_\p(1)\\
&=\frac{1}{\alpha^* N}\sum_{k=1}^N\id_{\{S^{(k)}>s\}} \int_{s}^{S^{(k)}} 1\d x -\frac{1}{\alpha^*}\int_{s}^{\infty} 1-G(x)\d x+\smallO_\p(1)\\
&=\frac{1}{\alpha^* N}\sum_{k=1}^N(S^{(k)}-s)_+-\frac{1}{\alpha^*}\E[(S-s)_+]+\smallO_\p(1).
\end{aligned}\end{equation}
Similarly, for $I_3$, we have 
\begin{align}\label{eq:I_3_ES}
I_3= \frac{1}{\alpha}\sum_{i=1}^n\int_{1-\alpha}^1 \left((\wF^N_i)^{-1}(p)-F_i^{-1}(p)\right)\d p=\frac{1}{\alpha N}\sum_{k=1}^N\sum_{i=1}^n (X_i^{(k)}-t_i)_+-\frac{1}{\alpha}\sum_{i=1}^n\E[(X_i-t_i)_+]+\smallO_\p(1).
\end{align}
Let $Z^{(k)}=(S^{(k)}-s)_+/\alpha^*-\sum_{i=1}^n(X_i^{(k)}-t_i)_+/\alpha$ for $k \in \N$. It is clear that $Z^{(1)}, Z^{(2)}, \dots$ are i.i.d.~with expected value 
$$\mu=\E[Z^{(k)}]=\frac{1}{\alpha^*}\E[(S-s)_+]-\frac{1}{\alpha}\sum_{i=1}^n\E[(X_i-t_i)_+],$$ 
and variance $\sigma^2:=\var(Z^{(k)})$. 
Define $\mathbf{A}_\ES \in \R^{n+1}$  by $\mathbf{A}_\ES=(1/\alpha, \dots, 1/\alpha, -1/\alpha^*)$ and a sequence of random vectors $\mathbf Y^{(k)}=( (X^{(k)}_1-t_1)_+, \dots, (X^{(k)}_n-t_n)_+, (S^{(k)}-s)_+)$ for $k \in \N$. It is clear that $Z^{(k)}=-\mathbf A^{\top} \mathbf Y^{(k)}$. Let $\Sigma_\ES$ 
be the covariance matrix of the random vector $\mathbf Y=( (X_1-t_1)_+, \dots, (X_n-t_n)_+, (S-s)_+)$. We have $\sigma^2=\mathbf A_\ES^\top \Sigma_\ES \mathbf A_\ES< \infty$ by the regularity condition in Assumption \ref{ass:DQES}.
Let $\bar Z=\sum_{k=1}^N Z^{(k)}/N$. By
 central limit theorem, we have 
\begin{align*}
\sqrt{N}(I_2-I_3)=\sqrt{N}(\bar Z-\mu)+\smallO_\p(1)\lawto \mathrm{N}(0, \sigma^2).
\end{align*}
Let $$c:=\frac{ \alpha^*s-\int^{1}_{1-\alpha^*}G^{-1}(p)\d p}{(\alpha^*)^2/\alpha}=\frac{\VaR_{\alpha^*}(S)-\ES_{\alpha^*}(S)}{\DQES(\bX)}.$$ Positive density of $G$ in Assumption \ref{ass:DQES} further implies $c<0$. Moreover, we have $I_1=c(\alpha^N-\alpha^*)/\alpha+\smallO_\p(1)$.
Combine the fact $I_1+I_2-I_3=0$, we have 
\begin{align*}
\sqrt{N}(\wDQES(N) -\DQES(\bX))=\frac{I_1}{c}+\smallO_\p(1)=\frac{\sqrt{N}(I_3-I_2)}{c} \lawto \mathrm{N}\left(0, \frac{\sigma^2}{c^2}\right),
\end{align*}
which completes the proof. 
\end{proof}

\subsection{Proofs in Section \ref{sec:5}}\label{app:depend}

We first present a lemma for uniformly convergence of empirical distribution with $\alpha$-mixing data from \cite[Corollary 2.1]{CR92}.

\begin{lemma}\label{lem:alpha-dis}
Let $\{X^{(k)}\}_{k \ge 1}$ be a stationary $\alpha$-mixing sequence of real-valued random variables with distribution $F$ and mixing coefficient $\alpha(m)$ satisfying 
\begin{equation}\label{eq:alpha}
\sum_{m=1}^ \infty m^{-1} (\log m) (\log \log m)^{1+\delta} \alpha(m) < \infty,~~~\mbox{for some}~~~\delta>0,
\end{equation}
and let $\wF^{N}$ be the empirical distribution function based on the $X^{(1)}, \dots, X^{(N)}$.
 Then, $$\sup_{x \in \R} \vert \wF(x) -F(x)\vert  \to 0, ~~~\p\text{-a.s.}$$
\end{lemma}

\begin{proof}[Proof of Proposition \ref{prop:alpha-con}]

From the proof of Theorem \ref{thm:1},  we have the convergence of $\DQrho$ holds if $X^{(k)}_i \lawto X_i$ for all $i \in [n]$ and $\wG \lawto G$. 

Since the mixing coefficient of $\{\bX^{(k)}\}_{k \ge 1}$  satisfies \eqref{eq:alpha}, the mixing coefficients of $\{X^{(k)}_i\}_{k \ge 1}$, $i\in [n]$ and $\{\sum_{i=1}^n X^{(k)}_i\}_{k \ge 1}$ also satisfies \eqref{eq:alpha}. Therefore, we have
 $\sup_{x \in \R} \vert \wF_i(x) -F_i(x)\vert  \to 0$, $\p$-a.s. for all $i \in [n]$ and $\sup_{x \in \R} \vert \widehat G(x) -G(x)\vert  \to 0$, $\p$-a.s.~by Lemma \ref{lem:alpha-dis}. 
 Hence, following the proof of Theorem \ref{thm:1} and the uniformly converges of empirical distributions for $X_i$ and $S$, we  can complete the proof. 
 \end{proof}
Recall that the three key tools we used in the proofs of Theorem \ref{thm:1} and \ref{thm:2} are the strong convergence of empirical process $E^N$ (Lemma \ref{lem:1}), Bahadur's representation (Lemma \ref{lem:Bahadur}) and central limit theorem. We first present similar lemmas for $\alpha$-mixing sequences.


We present the following strong approximation lemma for a stationary $\alpha$-mixing sequence in \cite{PP80}.
\begin{lemma}\label{lem:emp-alpha}
Let $\{U^{(k)}\}_{k\ge 1}$ be a stationary $\alpha$-mixing sequence of uniform distributed random variables
 satisfying $$\alpha(m)=\mathcal{O}(m^{-5-\epsilon})~~~\text{for some}~~~ 0<\epsilon\le\frac{1}{4}.$$
 Then, without changing its distribution, we can redefine the $\delta_N$ of $\{U^{(k)}\}_{k\ge 1}$ on a richer probability space on which there exists a Kiefer process $\{K(s,t); ~0\le s\le 1,~ t \ge 0\}$
 and a constant $\lambda>0$ depending only on $\epsilon$ such that $$\sup_{N \le T} \sup_{0\le s\le 1} \vert N^{1/2} \delta_N(s)-K(s,N) \vert =\O(T^{1/2} (\log(T))^{-\lambda}) ~~~\text{a.s.}$$
\end{lemma}
It is well known that the standardized process of the Kiefer process
$$B(s):=t^{-1/2}K(s, t), ~~~s\in [0,1],$$
is a Brownian bridge for every fixed $t>0$. Hence
Lemma \ref{lem:emp-alpha} implies that for stationary $\alpha$-mixing sequence with $\alpha(m) =\O(m)$, there exists a Browian bridge process such that
$$\sup_{y\in \mathbb{R}}|\delta^N(y)-B^N(y)|={\smallO}_{\p}(1).$$

\begin{lemma}\label{lem:quan-alpha}
Let $\delta>0$, $\beta>0$ and $p\in (0,1)$. Assume $\{X^{(k)}\}_{k \ge 1}$ is a second-order stationary $\alpha$-mixing sequence with the mixing coefficients such that  $\alpha(m)=\O(m^{-\beta})$ where $\beta>\max\{3+5/(1+\delta), 1+2/\delta\}$ and $\p(X_i=X_j)=0$ for any $i\neq j$. Assume that the common marginal distribution function $F(x)$ has a positive continuous density $f(x)$ where $f(x)$ and $f'(x)$ are bounded in some neighborhood of $F^{-1}(p)$. Then, as $N \to \infty$,
$$(\wF^N)^{-1}(p)=F^{-1}(p)+\frac{p-\wF^N(F^{-1}(p))}{f(F^{-1}(p))}+O\left(N^{-\frac{3}{4}+\frac{\delta}{4(2+\delta)}}(\log\log N\cdot\log N)^{\frac{1}{2}}\right), ~~~\text{a.s.}$$

\end{lemma}
Lemma \ref{lem:quan-alpha} implies 
$$(\wF^N)^{-1}(p)=F^{-1}(p)+\frac{p-\wF^N(F^{-1}(p))}{f(F^{-1}(p))}+\smallO_\p(1)$$
for any  second-order stationary $\alpha$-mixing sequence  $\{X^{(k)}\}_{k \ge 1}$ such that  $\alpha(m)=\O(m^{-\beta})$ with $\beta>3$.

Below is CLT for $\alpha$-mixing sequence from \cite{I62}
\begin{lemma}\label{lem:CLT-alpha}
Suppose $\{X^{(k)}\}_{k \ge 1}$ is an $\alpha$-mixing strictly stationary where $\E[X^{(k)}]=0$, $\E[\vert X^{(k)}\vert^{2+\delta}<\infty]$ and 
$
\sum_{m=1}^\infty (\alpha(m))^{\delta/(2+\delta)}< \infty$ 
for some $\delta>0$.
Let $S_N=\sum_{k=1}^N X^{(k)}$ for $N \ge 1$. Then 
$$\sigma^2=\var(X^{(1)})+2\sum_{k=2}^\infty \cov(X^{(1)}, X^{(k)})< \infty.$$
Moreover, if $\sigma >0$, then
$$\frac{S_N}{\sigma \sqrt{N}}\lawto \mathrm{N}(0,1), ~~~\text{as} ~~~N \to \infty.$$
\end{lemma}
If $\alpha(m)=\O(m^{-\epsilon})$ for some $\epsilon>0$, then we need select $\delta> 2/(\epsilon-1)$ in Lemma \ref{lem:CLT-alpha}.

\begin{proof}[Proof of Theorem \ref{thm:alphaDQVaR}]
  We follow the proof of  Proposition \ref{prop:3} to decompose 
$\wDQVaR(N)-\mathrm{DQ}^{\VaR}_{\alpha}(\mathbf{X})$ into 
$I_1+I_2$. For  $I_1$, using Lemma \ref{lem:emp-alpha}, we still have
\begin{align*}
 \sup_{y\in \mathbb{R}}|\gamma^N(y)-B^N(G(y))|=\smallO_\p(1),
\end{align*}
where $B^N$ is a standard Brownian bridge for each $N \in \N$. Since Assumption \ref{ass:DQVaR-alpha} is stronger than Assumption \ref{ass:DQVaR},  using the same argument in the  proof of  Proposition \ref{prop:3}, we can express $I_1$ as 
\begin{equation}\label{eq:I_11}
\sqrt{N}I_1=\sqrt{N}\left(  \wG^N\left( t_{n+1}\right)-G\left( t_{n+1}\right) \right)+\smallO_\p(1).
\end{equation}

For  $I_2$,  by  Lemma \ref{lem:quan-alpha} and  the mean value theorem, we have  
\begin{equation}\label{eq:I_22}
\sqrt{N}I_2=    \mathbf H^\top\mathbf M^N+\smallO_\p(1),
\end{equation}
where 
 $$\mathbf H=\left(\frac{g(t_{n+1})}{f_1(t_1)}, \dots, \frac{g(t_{n+1})}{f_n(t_n)}\right) ~~~\mbox{and}~~~{\mathbf M}^N=\sqrt{N}\left(F_1(t_1)-{\wF}^{N}_1(t_1),\dots,F_n(t_n))-{\wF}^{N}_n(t_n)\right) .
$$

For $k \ge 1$,  let  $\mathbf Z^\k=(Z^\k_1, \dots, Z^\k_{n+1})$ where $Z^\k_i=F_i(t_i)-\mathbb{I}_{\{X^\k_i\leq t_i\}}$ for $i\in [n]$ and $Z^\k_{n+1}=G(t_{n+1})-\mathbb{I}_{\{\sum_{i=1}^n X^\k_i\leq t_{n+1}\}}$. 
Combining with \eqref{eq:I_11} and \eqref{eq:I_22}, we have
$$\sqrt{N}\left(\wDQVaR(N)-\mathrm{DQ}^{\VaR}_\alpha(\mathbf X)\right)=\sqrt{N}\frac{I_1+I_2}{\alpha}=\frac{\sum_{k=1}^N\mathbf A_\VaR^\top{\mathbf Z}^{(k)}}{\sqrt{N}}+\smallO_\p(1),$$ where $\mathbf A_\VaR=\left(\mathbf H/\alpha, -1/\alpha\right)$.
It is clear that $\{\mathbf A_{\VaR}^\top\mathbf Z^\k\}_{k \ge 1}$ is a stationary $\alpha$-mixing sequence with $\E[\mathbf A_{\VaR}^\top\mathbf Z^\k]=0$.  By Assumption \ref{ass:alpha},  $\alpha(\mathbf A_{\VaR}^\top\mathbf Z^\k,m)\le \alpha(\{\mathbf X^{(k)}\}_{k \ge 1})=\O(m^{-5-\epsilon})$. Hence, for $\delta>1/2$,  $\sum_{m \ge 1} (\alpha(m))^{\delta/(2+\delta)}< \infty$. By Assumption \ref{ass:DQVaR-alpha}, $\mathbf A_{\VaR}^\top\mathbf Z^\k$ is bounded for all $k \ge 1$; hence, $\E[\vert \mathbf A_{\VaR}^\top\mathbf Z^\k\vert ^{2+\delta}]< \infty$. Therefore, 
$\{\mathbf A_{\VaR}^\top\mathbf Z^\k\}_{k \ge 1}$ satisfies conditions in Lemma  \ref{lem:CLT-alpha}. Hence, by Lemma \ref{lem:CLT-alpha}, we have 
$$\sqrt{N}\left(\wDQVaR(N)-\mathrm{DQ}^{\VaR}_\alpha(\mathbf X)\right)\lawto {\rm N}(0,\sigma_\VaR^2),$$
where 
\begin{align*}
\sigma^2_\VaR&=\var(\mathbf A_{\VaR}^\top\mathbf Z^{(1)})+2\sum_{k=2}^\infty \cov(\mathbf A_{\VaR}^\top\mathbf Z^{(1)}, \mathbf A_{\VaR}^\top\mathbf Z^\k)\\
&=\mathbf A_{\VaR}^\top\left(\var(\mathbf Z^{(1)})+2\sum_{k=2}^\infty\cov(\mathbf Z^{(1)},\mathbf Z^\k)\right) \mathbf A_{\VaR}\\
&=\mathbf A_{\VaR}^\top\left(\var(\mathbf I^{(1)})+2\sum_{k=2}^\infty\cov(\mathbf I^{(1)},\mathbf I^\k)\right) \mathbf A_{\VaR},
\end{align*}
with $\mathbf I^{(k)}=(\id_{\{X_1^{(k)}\le t_1\}}, \dots, \id_{\{X_n^{(k)}\le t_n\}})$. Hence, we have  $\sigma_\VaR=\mathbf A_{\VaR}^\top\Sigma^\alpha_\VaR\mathbf A_{\VaR}$ with $\Sigma^\alpha_{\VaR}=\var(\mathbf I^{(1)})+2\sum_{k=2}^\infty\cov(\mathbf I^{(1)},\mathbf I^\k)$.
\end{proof}

\begin{proof}[Proof of Theorem \ref{thm:alphaDQES}]

Define $\mathbf{A}_\ES \in \R^{n+1}$  by $\mathbf{A}_\ES=(1/\alpha, \dots, 1/\alpha, -1/\alpha^*)$, a sequence of random vectors $\mathbf Y^{(k)}=( (X^{(k)}_1-t_1)_+, \dots, (X^{(k)}_n-t_n)_+, (S^{(k)}-s)_+)$ for $k\ge 1$ and $$c:=\frac{ \alpha^*s-\int^{1}_{1-\alpha^*}G^{-1}(p)\d p}{(\alpha^*)^2/\alpha}=\frac{\VaR_{\alpha^*}(S)-\ES_{\alpha^*}(S)}{\DQES(\bX)}.$$
Following the discussion in the proof of Theorem \ref{thm:2}, we have 
\begin{align*}
\sqrt{N}(\wDQES(N) -\DQES(\bX))=-\frac{\sum_{k=1}^N \left(\mathbf A^{\top} \mathbf Y^{(k)}-\E[\mathbf A^{\top} \mathbf Y^{(k)}]\right)}{c\sqrt{N}}.
\end{align*}

Let  $Z^{(k)}=\mathbf A^{\top} \mathbf Y^{(k)}-\E[\mathbf A^{\top} \mathbf Y^{(k)}]$ for all ${k \ge 1}$. It is clear that $\{Z^\k\}_{k \ge 1}$ is a stationary $\alpha$-mixing sequence with $\E[Z^\k]=0$.  By Assumption \ref{ass:alpha},  $\alpha_{\{Z^\k\}_{k \ge 1}}(m)\le \alpha_{\{\mathbf X^{(k)}\}_{k \ge 1}}(m)=\O(m^{-5-\epsilon})$. 
Hence, for $\delta>1/2$, we have $\sum_{m =1}^{\infty} (\alpha_{\{Z^\k\}_{k \ge 1}}( m))^{\delta/(2+\delta)}< \infty$. 
By Assumption \ref{ass:DQES-alpha}, there exists some $\delta>1/2$  such that $\E[(X_i-t_i)_+^2+\delta]< \infty$ for all $i \in [n]$ and $\E[(S-s)_+^2+\delta]< \infty$. Hence , $\E[\vert Z^\k\vert ]< \infty$. Therefore, $\{Z^k\}_{k \ge 1}$
satisfies the condition of Lemma \ref{lem:CLT-alpha}. 
Hence, 
\begin{align*}
\sqrt{N}(\wDQES(N) -\DQES(\bX))\lawto \mathrm{N}(0, \sigma_\ES^2),
\end{align*}
where 
\begin{align*}
\sigma_\ES^2&=\frac{1}{c^2}\var(\mathbf A^{\top} \mathbf Y^{(1)})+2\sum_{k=2}^\infty \cov(\mathbf A^{\top} \mathbf Y^{(1)}, \mathbf A^{\top} \mathbf Y^{(k)})\\
&=\frac{1}{c^2}\mathbf A^{\top}\left(\var( \mathbf Y^{(1)})+2\sum_{k=2}^\infty \cov( \mathbf Y^{(1)}, \mathbf Y^{(k)})\right)\mathbf A,
\end{align*}
which completes the proof.

\end{proof}

\subsection{Proof in Section \ref{sec:6}}

\begin{proof}[Proof of Theorem \ref{prop:5}]
\begin{itemize}
    \item[(i)] 
Note that  
\begin{equation*}\sqrt{N}\left(\widehat{\mathrm{DR}^{\VaR}}(N)-\mathrm{DR}^{\VaR}(\mathbf{X})\right)=\sqrt{N}\left(\frac{({\wG^N})^{-1}(1-\alpha)}{\sum_{i=1}^{n}(\wF^N_i)^{-1}(1-\alpha)}-\frac{G^{-1}(1-\alpha)}{\sum_{i=1}^{n}F_i^{-1}(1-\alpha)}\right).
    \end{equation*}
    By Lemma \ref{lem:Bahadur}, we have   $$ (\wF_i^N)^{-1}(1-\alpha)-F_i^{-1}(1 -\alpha) =\frac{(1-\alpha)-{\wF_i^N}(F_i^{-1}(1-\alpha))}{f_i(F_i^{-1}(1-\alpha))} +\smallO_\p(1/\sqrt{N}),$$    and  $$ (\wG^N)^{-1}(1-\alpha)-G^{-1}(1-\alpha) =\frac{(1-\alpha)-{\wG^N}(G^{-1}(1-\alpha))}{g(G^{-1}(1-\alpha))} +\smallO_\p(1/\sqrt{N}).$$

 By  Delta method, we have 
 $$ \sqrt{N}\left(\widehat{\mathrm{DR}^{\VaR_{\alpha}}}(N)-\mathrm{DR}^{\VaR_{\alpha}}(\mathbf{X})\right)=\textbf{R}^\top_{\VaR} \textbf{L}_{\VaR}^N+\smallO_\p(1),$$
    where  
    $$\mathbf{R}_{\VaR}=\left(\frac{s}{f_1(t_1)t_{n+1}^2}, \dots, \frac{s}{f_n(t_n) t_{n+1}^2}, \frac{-1}{g(s)t_{n+1}}\right),$$ and $$\textbf{L}_{\VaR}^N=\sqrt{N}\left((1-\alpha)-\wF^N_1(t_1),\dots,(1-\alpha)-\wF^N_n(t_n),(1-\alpha)- {\wG^N}(s)\right).$$

Along the similar lines as in Theorem \ref{thm:2}, for $k \in \N$,  let  $\mathbf Z^\k=(Z^\k_1, \dots, Z^\k_{n+1})$ where $Z^\k_i=F_i(t_i)-\mathbb{I}_{\{X^\k_i\leq t_i\}}$ for $i\in [n]$ and $Z^\k_{n+1}=G(s)-\mathbb{I}_{\{\sum_{i=1}^n X^\k_i\leq s\}}$.  
We have  $$ \sqrt{N}\overline{\mathbf Z}^N:=\sqrt{N}\frac{\sum_{k=1}^N \mathbf Z^\k}{N}=\mathbf L_{\VaR}^N. $$
By multiple central limit theorem, we  obtain $ \overline{\mathbf Z}^N\lawto {\rm N}(0,\Sigma_{\VaR}), $
where $\Sigma_\VaR$ is the covariance matrix of $\mathbf Z^\k$. Therefore, $$\sqrt{N}\left(\widehat{\mathrm{DR}^{\VaR_{\alpha}}}(N)-\mathrm{DR}^{\VaR_{\alpha}}(\mathbf{X})\right)\lawto \mathrm{N}(0,\mathbf R_{\VaR}^\top\Sigma_{\VaR} \mathbf R_{\VaR}).$$
    \item[(ii)] Note that $$    \sqrt{N}\left(\widehat{\mathrm{DR}^{\ES}}(N)-\mathrm{DR}^{\ES}(\mathbf{X})\right)=\sqrt{N}\left(\frac{\frac{1}{\alpha}\int_{1-\alpha}^1 (\wG^N)^{-1}(p)\d p}{\frac{1}{\alpha} \sum_{i=1}^{n}\int_{1-\alpha}^1 (\wF_i^N)^{-1}(p)\d p}-\frac{\frac{1}{\alpha} \int_{1-\alpha}^1 G^{-1}(p)\d p}{\frac{1}{\alpha} \sum_{i=1}^{n}\int_{1-\alpha}^1 F_i^{-1}(p)\d p}\right).
   $$ 

   By Lemma \ref{lem:Bahadur}, and in conjunction with \eqref{eq:I_2_ES} and \eqref{eq:I_3_ES}, we obtain
\begin{align*}
\int_{1-\alpha}^1 \left((\wF^N_i)^{-1}(p)-F_i^{-1}(p)\right)\d p=\frac{1}{ N}\sum_{k=1}^N  (X_i^{(k)}-t_i)_+-\E[(X_i-t_i)_+]+\smallO_\p(1),
\end{align*}
and  
\begin{align*}
 \int_{1-\alpha}^{1} \left( (\wG^N)^{-1}(p)-G^{-1}(p)\right) \d p
&=\frac{1}{ N}\sum_{k=1}^N(S^{(k)}-s)_+-\E[(S-s)_+]+\smallO_\p(1).
\end{align*} 
Again, using Delta method, we have 
 $$ \sqrt{N}\left(\widehat{\mathrm{DR}^{\ES_{\alpha}}}(N)-\mathrm{DR}^{\ES_{\alpha}}(\mathbf{X})\right)=\textbf{R}^\top_{\ES} \textbf{L}_{\ES}^N+\smallO_\p(1),$$
    where   $$\mathbf{R}_{\ES}=\frac{1}{\alpha}\left(-\frac{\ES_\alpha(S)}{\left(\sum_{i=1}^n\ES_\alpha(X_i)\right)^2}, \dots, -\frac{\ES_\alpha(S)}{\left(\sum_{i=1}^n\ES_\alpha(X_i)\right)^2},  \frac{1}{\sum_{i=1}^n\ES_\alpha(X_i)}\right),$$ and 
 $$\begin{aligned}\mathbf L_{\ES}^N=\sqrt{N}&\left( \int_{1-\alpha}^1 \left((\wF^N_1)^{-1}(p)-F_1^{-1}(p)\right)\d p,\dots,\int_{1-\alpha}^1 \left((\wF^N_n)^{-1}(p)-F_n^{-1}(p)\right)\d p, \right.\\&\left.\int_{1-\alpha}^1 \left((\wG^N)^{-1}(p)-G^{-1}(p)\right)\d p\right).\end{aligned}$$
For $k \in \N$,  let  $\mathbf Z^\k=(Z^\k_1, \dots, Z^\k_{n+1})$ where $Z^\k_i= (X_i^{(k)}-t_i)_+-\E[(X_i-t_i)_+]$ for $i\in [n]$ and $Z^\k_{n+1}=(S^{(k)}-s)_+-\E[(S-s)_+]$.
We have  $$ \sqrt{N}\overline{\mathbf Z}^N:=\sqrt{N}\frac{\sum_{k=1}^N \mathbf Z^\k}{N}=\mathbf L_{\ES}^N. $$

By multiple central limit theorem, we  obtain $ \overline{\mathbf Z}^N\lawto {\rm N}(0,\Sigma_{\ES}), $
where $\Sigma_\ES$ is the covariance matrix of  $( (X_1-t_1)_+, \dots, (X_n-t_n)_+,(S-s)_+)$.  
Therefore, $$\sqrt{N}\left(\widehat{\mathrm{DR}^{\ES_{\alpha}}}(N)-\mathrm{DR}^{\ES_{\alpha}}(\mathbf{X})\right)\lawto \mathrm{N}(0,\mathbf R^\top_{\ES}\Sigma_{\ES} \mathbf R_{\ES}).$$\end{itemize}  \end{proof}

\end{document}